\documentclass[lettersize,journal]{IEEEtran}
\usepackage{amsmath,amsfonts}
\usepackage{algorithmic}
\usepackage{algorithm}
\usepackage{array}
\usepackage[caption=false,font=normalsize,labelfont=sf,textfont=sf]{subfig}
\usepackage{textcomp}
\usepackage{stfloats}
\usepackage{url}
\usepackage{verbatim}
\usepackage{graphicx}
\usepackage{cite}
\usepackage{comment}
\usepackage{amssymb}
\usepackage{amsthm}
\usepackage{braket}
\usepackage{mathtools}
\usepackage{xcolor}

\usepackage[colorinlistoftodos,textwidth=\marginparwidth]{todonotes} % 

\newtheorem{theorem}{Theorem}
\newtheorem{lemma}{Lemma}
\newtheorem{corollary}{Corollary}

\newtheorem{proposition}{Proposition}
\newtheorem{remark}{Remark}
\begin{document}

\title{Robust Covert Quantum Communication\\ under Bounded Channel Uncertainty%\thanks{Code will be made available upon acceptance.}
}

\author{Abbas Arghavani, Alessandro V. Papadopoulos, Vahid Azimi Mousolou, Giuseppe Nebbione, Shahid Raza
        % <-this % stops a space
\thanks{This project was fully funded by M\"{a}lardalen University through its strategic profile area Trusted Smart Systems (TSS).}
\thanks{A. Arghavani, A.V. Papadopoulos, and G. Nebbione are with the Department of Computer Science and Engineering, M{\"a}lardalen University, V{\"a}ster{\aa}s, Sweden (email: abbas.arghavani@mdu.se, alessandro.papadopoulos@mdu.se, giuseppe.nebbione@mdu.se)}
%\thanks{Alessandro V. Papadopoulos is with the Department of Computer Science and Engineering, M{\"a}lardalen University, V{\"a}ster{\aa}s, Sweden (email: alessandro papadopoulos@mdu.se)}
\thanks{Vahid Azimi Mousolou is with the Department of Physics and Astronomy, Uppsala University, Uppsala, Sweden (email: vahid.azimi-mousolou@physics.uu.se)}
%\thanks{Giuseppe Nebbione is with the Department of Computer Science and Engineering, M{\"a}lardalen University, V{\"a}ster{\aa}s, Sweden (email: giuseppe.nebbione@mdu.se)}
\thanks{Shahid Raza is with the School of Computing Science, University of Glasgow, Glasgow, UK, and also with the Department of Computer Science and Engineering, M\"alardalen University, V\"aster{\aa}s, Sweden (email: shahid.raza@glasgow.ac.uk).}
}

% The paper headers
%\markboth{Journal of \LaTeX\ Class Files,~Vol.~14, No.~8, August~2021}%
%{Shell \MakeLowercase{\textit{et al.}}: A Sample Article Using IEEEtran.cls for IEEE Journals}

%\IEEEpubid{0000--0000/00\$00.00~\copyright~2021 IEEE}
% Remember, if you use this you must call \IEEEpubidadjcol in the second
% column for its text to clear the IEEEpubid mark.

\maketitle

\begin{abstract}
Existing theoretical limits for covert quantum communication are typically derived under idealized assumptions, namely, that channel parameters, such as transmissivity and background noise, are perfectly known and constant. However, real-world optical links, including satellite, fiber, and free-space systems, are subject to dynamic environmental conditions, calibration noise, and estimation errors that violate these assumptions. This paper addresses this gap by analyzing covert quantum communication over \emph{compound quantum optical channels} with \emph{bounded} uncertainty in both transmissivity and thermal noise, enabling worst-case guarantees for security-critical operation.

We develop a robust analytical framework that guarantees security and reliability across all admissible channel realizations. Crucially, robustness is \emph{not} obtained by naively substituting worst-case parameter values into known-channel covert bounds, because the channel realizations that extremize covertness and reliability occur at different corners of the uncertainty set.

A key insight is that the worst-case conditions for covertness and reliability are driven by opposite assumptions on environmental noise, even though both are governed by the same worst-case transmittance corner, leading to a fundamental trade-off that must be resolved to certify both stealth and decodability for secure system design. We derive a closed-form guaranteed worst-case lower bound on the expected number of covert qubits that can be reliably transmitted despite this uncertainty. Our analysis further reveals
a sharp feasibility boundary, a rate cliff edge, beyond which
the worst-case guaranteed covert payload collapses to zero,
and we quantify the corresponding security tax induced by the mismatch between the worst-case covertness and worst-case reliability extremizers under bounded uncertainty. We validate the Willie-side covertness component via QuTiP simulations of the underlying four-mode bosonic model, including Fock-space cutoff convergence to ensure numerical fidelity, and combine the resulting numerically estimated covertness constant with the analytical hashing-bound reliability model to evaluate the robust payload. Overall, our results move covert quantum communication from nominal, perfect-knowledge analyses to certified worst-case operation under uncertainty, opening the door to secure applications in adversarial and uncertain environments.
\end{abstract}
\vspace{-3mm}
\begin{center}
\footnotesize
\textit{This work has been submitted to the IEEE for possible publication.
Copyright may be transferred without notice, after which this version may no longer be accessible.}
\end{center}
\vspace{-5mm}
\section{Introduction}
\vspace{-1mm}
While cryptography protects the content of communication, covert communication (also known as Low Probability of Detection (LPD) communication) aims to conceal the existence of transmission itself~\cite{bash2013limits}. The fundamental limit of this paradigm is governed by the \emph{square-root law} (SRL), which asserts that only $\mathcal{O}(\sqrt{n})$ covert bits can be reliably transmitted over $n$ uses of a classical or quantum channel~\cite{bash2013limits, arghavani2023covert}. Extending these principles into the quantum domain has become critical for emerging security applications such as covert satellite links, secure financial transactions, and undetectable command-and-control in tactical environments~\cite{anderson2024covert, anderson2024square, anderson2025achievability, tahmasbi2021signaling}. 

In covert quantum communication, a legitimate transmitter (Alice) seeks to deliver information to a receiver (Bob) while ensuring that an adversarial warden (Willie) cannot detect that communication is occurring. As illustrated in Fig.~\ref{fig:simple_model}, this scenario can be modeled as an optical channel in which a fraction of the transmitted photons is inevitably lost to the environment. Willie can intercept and analyze these lost photons to infer Alice’s activity. This abstraction, detailed further in Fig.~\ref{fig:channel}, captures the essential detection challenge that covert communication seeks to overcome.

\begin{figure}[!t]
    \centering
    \includegraphics[width=0.5\linewidth]{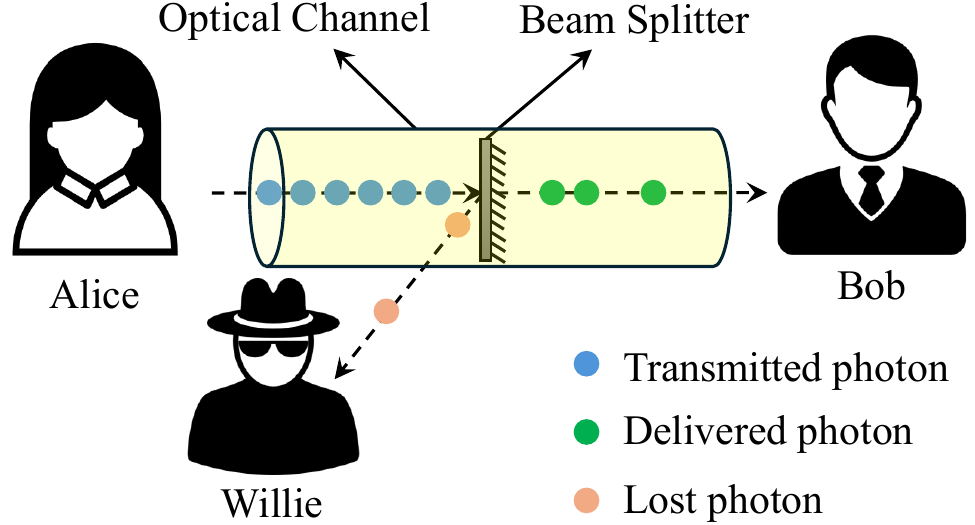}
    %\Description{High-level diagram of covert quantum communication showing Alice sending photons through a beam splitter channel with transmittance $\eta$. Bob receives a fraction of the photons while the rest reach the adversary Willie, illustrating the goal of keeping Willie’s observations indistinguishable from background noise.}
\caption{\textbf{Covert quantum communication schematic.} Beamsplitter channel with transmittance $\eta$: Bob gets $\eta$, Willie gets $1-\eta$. Alice enforces covertness by making Willie’s observations indistinguishable from thermal noise.}
    \label{fig:simple_model}
    \vspace{-5mm}
\end{figure}

Despite recent theoretical advances, existing quantum covert communication models rely on a key idealization: they assume that the channel parameters (transmittance ($\eta$) and thermal background noise ($\overline{n}_B$)) are perfectly known and remain constant. This assumption does not hold in realistic deployments such as satellite or fiber-optic links, where these parameters fluctuate due to alignment drift, atmospheric turbulence, and temperature variation. Consequently, systems designed under perfect-knowledge assumptions may fail to guarantee covertness or reliability once uncertainties are introduced. While probabilistic uncertainty models are useful for average-case analysis, they are less suitable for high-assurance covert communication, where even a single covertness failure may be unacceptable. We therefore adopt a bounded uncertainty model, which enables worst-case guarantees uniformly over all admissible channel realizations.

This paper addresses this gap by developing a robust analytical framework for \emph{covert quantum communication under bounded channel uncertainty}. We characterize the guaranteed covert payload when the channel parameters are not perfectly known. Let $M(n)$ denote the number of covert qubits reliably delivered over $n$ channel uses under a prescribed covertness constraint. Because Alice uses randomized transmission, $M(n)$ is a random block-level quantity, and $\mathbb{E}[M(n)]$ denotes the expected covert payload. Under bounded uncertainty, its guaranteed worst-case counterpart is denoted by $M_{\mathrm{rob}}(n)$.

We reveal and resolve a fundamental conflict: covertness is governed by channel conditions most favorable to Willie, whereas reliability is governed by conditions most adverse to Bob. This misalignment creates a nontrivial robust-design trade-off that is absent from perfect-knowledge models. Our analysis further shows that uncertainty does not merely reduce the guaranteed covert payload gradually; beyond a critical threshold, the worst-case guaranteed payload collapses abruptly to zero. We refer to this sharp feasibility boundary as the “rate cliff edge” and quantify the resulting “security tax” caused by the mismatch between the worst-case covertness and worst-case reliability extremizers.

A natural question is whether robustness can be obtained simply by substituting worst-case parameter values into known-channel bounds. Our analysis shows that this is not sufficient: the realizations that are worst for covertness and those that are worst for reliability generally occur at different corners of the uncertainty set. Hence, robust design requires enforcing both constraints jointly over the full compound set.

Our main contributions are summarized as follows:
\begin{itemize}
    \item \textbf{Compound-channel formulation for covert quantum links:}
    We model covert quantum communication over optical channels as a compound quantum channel with bounded uncertainty in transmittance $\eta$ and background noise $\bar{n}_B$, enabling uniform worst-case guarantees over all admissible channel realizations.

    \item \textbf{Conflict between covertness and reliability:}
    We show that the channel conditions most detrimental to covertness and reliability occur at different boundaries of the uncertainty set, so robust design cannot be obtained by a naive worst-case substitution.

    \item \textbf{Rate cliff edge and robust payload bound:}
    By jointly enforcing covertness and reliability over the compound set, we derive a closed-form worst-case guaranteed lower bound on the expected covert payload and identify a sharp feasibility boundary beyond which guaranteed covert communication is impossible.

    \item \textbf{Quantification of the security tax:} Using MATLAB and QuTiP~\cite{qutip2}, we numerically validate the Willie-side covertness component of the analysis and quantify the payload loss caused by the mismatch between the covertness and reliability extremizers under bounded uncertainty.
\end{itemize}

These results bridge the gap between idealized covert quantum models and robust design under realistic channel uncertainty. To our knowledge, this is the first rigorous derivation of a worst-case guaranteed lower bound on covert quantum transmission over a compound optical channel. The framework is directly relevant to quantum-secure networks, satellite links, and stealth optical systems operating under fluctuating environmental conditions. The parameter regimes analyzed in this paper (e.g., transmittance $\eta \approx 0.9$ and mean background noise $\overline{n}_B \approx 0.1$) correspond to realistic operating points observed in free-space optical links under nighttime conditions~\cite{bourgoin2013comprehensive} and standard telecom fibers.

Section~\ref{sec:related_work} reviews foundational results in classical and quantum covert communication and outlines the critical gap addressed here. Section~\ref{sec:model} introduces the system model and performance metrics.
%, followed by the formal threat model in Section~\ref{sec:threat_model}.
Section~\ref{sec:robust} develops the analytical framework and derives the robust bound, while Section~\ref{sec:discussion} presents numerical validation and trade-off analysis. Section~\ref{sec:discussion_implications} discusses implications for system design and highlights the impact of uncertainty on performance. Finally, Section~\ref{sec:conclusion} concludes and describes future research directions.

\section{Related Work}
\label{sec:related_work}

Our work builds upon literature in classical and quantum information theory. In this section, we contextualize our contribution by reviewing the foundations of covert communication, its extension to the quantum domain, and the unresolved challenge of channel uncertainty.
\vspace{-3mm}
\subsection{Foundations of Classical Covert Communication}

Covert communication was first rigorously examined in classical channels, where its fundamental limits were formally characterized. The foundational principle is the square-root law (SRL)~\cite{bash2013limits}, which states that, for a channel with known noise characteristics, Alice can reliably and covertly transmit only $\mathcal{O}(\sqrt{n})$ bits over $n$ channel uses. This constraint arises because the energy of Alice’s signal must remain low enough to be statistically indistinguishable from background noise fluctuations, which scale as $\sqrt{n}$. Seminal works established the SRL for various classical channels~\cite{bash2013limits}. Later studies demonstrated that the SRL can be exceeded, and a positive covert rate achieved, when Willie is uncertain about the channel noise.

This highlights the dual role of uncertainty in covert communication: while it can be exploited to exceed fundamental limits in classical settings, it also complicates reliable guarantees. This tension also emerges in practice. For instance, recent work~\cite{padmal2025fat} investigates in-body covert communication using fat tissue, combining theoretical analysis with experimental validation. It demonstrates that covertness is achievable even in highly lossy and size-constrained settings, and emphasizes the real-world impact of physical-layer channel uncertainty, which we address in the quantum domain.

Our work revisits this tension in the quantum setting by analyzing how bounded uncertainty impacts covertness and reliability under worst-case assumptions. Related classical work has also explored cooperative and game-theoretic covert communication strategies, including dynamic role-switching between transmitting and jamming nodes~\cite{arghavani2024dynamic} and multi-warden covert communication models~\cite{arghavani2021game}, although these directions are outside the scope of the present work.
\vspace{-3mm}
\subsection{Classical-Quantum and Fully Quantum Extensions}%Extension to Classical-Quantum and Fully Quantum Channels

The extension of covert communication to the quantum realm began with classical-quantum (c-q) channels, where a classical input from Alice produces a quantum state at the receiver. This model captures quantum effects while retaining classical signaling. Early work extended the SRL and proved achievability for product-state inputs, though the converse remained open~\cite{sheikholeslami2016covert}. Later analysis characterized the fundamental limits and required secret key sizes~\cite{bullock2025fundamental}.

More recently, the analysis has advanced to fully quantum channels, where Alice transmits quantum states (qubits) rather than classical bits. Foundational work on the bosonic channel (the standard model for optical communication) was conducted by Anderson et al.~\cite{anderson2024square}, who demonstrated that the SRL also applies to covert qubit transmission using dual-rail photonic encoding. Their follow-up study derived key analytical expressions for the performance metrics central to our analysis~\cite{anderson2024covert}. Subsequently, this framework was extended to prove the achievability of SRL for arbitrary quantum channels~\cite{anderson2025achievability}. 
\vspace{-5mm}
\subsection{The Research Gap: Channel Parameter Uncertainty}
A common assumption in the covert quantum communication literature~\cite{sheikholeslami2016covert, anderson2024square, anderson2025achievability} is that Alice possesses perfect knowledge of the channel state information (CSI), and that the CSI (specifically, $\eta$ and $\bar{n}_B$) remains constant. This is a strong idealization that rarely holds in practice, where environmental variability and imperfect sensing introduce uncertainty. Although the impact of CSI uncertainty has been extensively studied in \textit{classical} covert communication, particularly in satellite systems requiring robust transmission design~\cite{jia2025robust}, a corresponding robust analysis of the \textit{quantum} SRL has remained an open and critical gap. This challenge has already been addressed in the classical domain under practical satellite communication scenarios, where the impact of dual-CSI uncertainty (arising from imperfect knowledge of both legitimate and adversarial channels) was modeled and mitigated using robust beamforming techniques~\cite{jia2025robust}. Their results highlight the need for resilience against environmental uncertainty in real-world covert systems. Our work brings this robustness perspective to the quantum setting, addressing the analogous challenge for covert optical communication under bounded channel uncertainty. By introducing a bounded uncertainty model, we extend the existing quantum framework to a more realistic setting and reveal the fundamental trade-off between guaranteeing covertness and ensuring reliability, a tension that remains hidden in models assuming perfect channel knowledge. This robust extension paves the way for covert quantum systems that are both theoretically sound and practically resilient.

From an information-theoretic perspective, our model of bounded uncertainty transforms the communication problem into one of determining robust performance over a \textit{compound quantum channel}. A compound channel is a model in which the actual channel is unknown but known to belong to a specified set, and a single, universal coding scheme must guarantee reliable communication for all channel realizations within that set. The classical capacity of such channels was formally established by Bjelaković and Boche~\cite{Bjelakovic2009}, who developed a framework for characterizing achievable rates in the absence of covertness constraints.

\section{System Model and Performance Metrics}
\label{sec:model}
\begin{figure*}[b!]
    \centering    \includegraphics[width=0.7\linewidth]{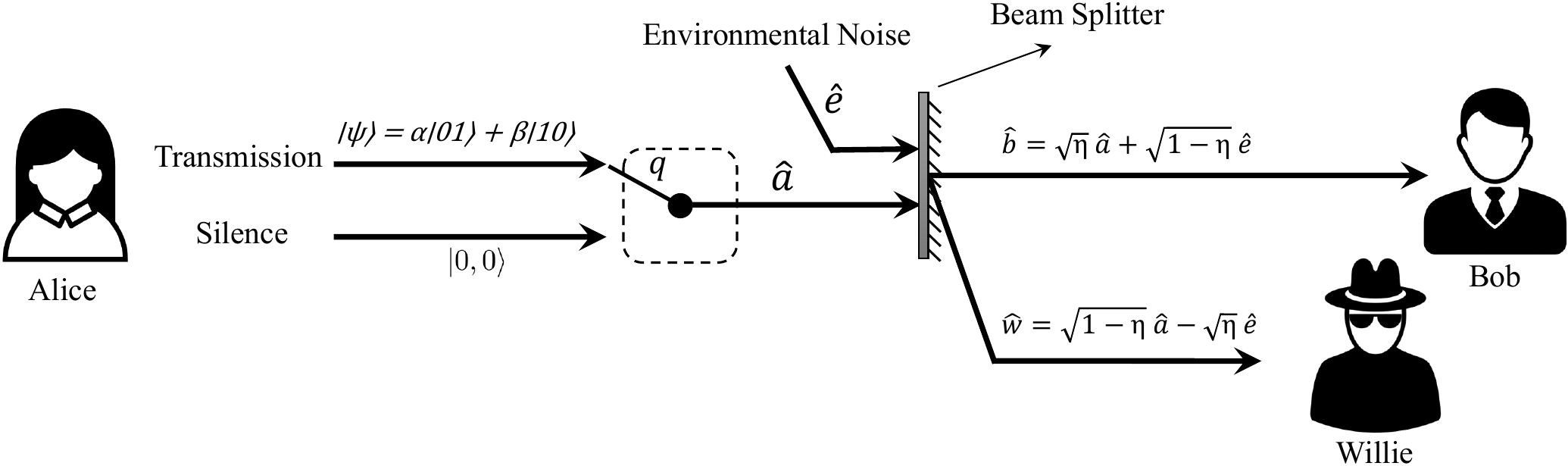}
    %\Description{Diagram of the covert quantum communication channel. Alice sends dual-rail single-photon qubits through a lossy thermal-noise bosonic channel modeled by a beamsplitter, producing Bob’s received mode and Willie’s intercepted mode according to the transmittance parameter $\eta$.}
    \caption{\textbf{Covert quantum channel model.} Alice sends a dual-rail qubit $\ket{\psi}=\alpha\ket{01}+\beta\ket{10}$ with probability $q$ (else vacuum $\ket{00}$). A beamsplitter mixes input $\hat a$ with thermal noise $\hat e$, yielding $\hat b=\sqrt{\eta}\hat a+\sqrt{1-\eta}\hat e$ (Bob) and $\hat w=\sqrt{1-\eta}\hat a-\sqrt{\eta}\hat e$ (Willie), with transmittance $\eta$.
    %\textbf{Channel model for covert quantum communication.} Alice transmits single-photon qubits, encoded via the dual-rail scheme, through a lossy thermal-noise bosonic channel modeled by a beamsplitter. With probability $q$, she sends $\ket{\psi} = \alpha\ket{01} + \beta\ket{10}$; otherwise, she remains silent and transmits $\ket{00}$. The beamsplitter mixes Alice’s input $\hat{a}$ with environmental noise $\hat{e}$, producing outputs $\hat{b} = \sqrt{\eta}\hat{a} + \sqrt{1-\eta}\hat{e}$ (received by Bob) and $\hat{w} = \sqrt{1-\eta}\hat{a} - \sqrt{\eta}\hat{e}$ (observed by Willie), where $\eta$ is the channel transmittance.
    }
    \label{fig:channel}
\end{figure*}
To analyze the impact of channel uncertainty, we first review the system model and performance metrics developed by Anderson et al.~\cite{anderson2024covert, anderson2024square}. These form the foundation of our robust formulation and are recast here to ensure self-containment and accessibility.
\vspace{-3mm}
\subsection{Physical System and Channel Model}\label{sec:model_channel}

We consider a covert quantum communication scenario involving three parties: a sender (Alice), a legitimate receiver (Bob), and a passive warden (Willie). Alice aims to transmit quantum information to Bob while ensuring that Willie cannot detect whether any communication is taking place. Time is divided into frames. In each frame, the optical channel is used $n$ times. We refer to each channel use as a \emph{slot}. In slot $t\in\{1,\ldots,n\}$, Alice decides probabilistically whether to transmit a qubit: with probability $q$, she sends a signal state; with probability $1 - q$, she transmits the vacuum state, $\ket{00}$ (i.e., remains silent). The signal state is a single qubit, physically implemented using dual-rail encoding, a widely used technique in linear optical quantum computing~\cite{anderson2024covert}. In this scheme, a logical qubit is represented by a single photon distributed across two distinct optical modes. The logical basis states, denoted with the subscript $L$, are defined as: $\ket{0}_L = \ket{01}$, and $\ket{1}_L = \ket{10}$, where $\ket{01}$ and $\ket{10}$ are single-photon Fock states. A general qubit is then expressed as a superposition:
\[
\ket{\psi} = \alpha \ket{0}_L + \beta \ket{1}_L = \alpha \ket{01} + \beta \ket{10},
\]
with $|\alpha|^2 + |\beta|^2 = 1$.

The physical transmission channel is modeled as a \emph{lossy thermal-noise bosonic channel}, which captures free-space and fiber-optic quantum links. As illustrated in Fig.~\ref{fig:channel}, this is implemented using a beamsplitter with two input ports and two output ports. One input port receives Alice's signal mode ($\hat{a}$), while the other receives an environmental noise mode ($\hat{e}$) modeled as a thermal state with mean photon number $\bar{n}_B$. The beamsplitter divides the signal such that one output mode reaches Bob ($\hat{b}$) and the other reaches Willie ($\hat{w}$). The input-output relations for the modal annihilation operators are:
\begin{align}
    \hat{b} &= \sqrt{\eta} \hat{a} + \sqrt{1-\eta} \hat{e}, \\
    \hat{w} &= \sqrt{1-\eta} \hat{a} - \sqrt{\eta} \hat{e},
\end{align}
where $\eta \in (0,1)$ is the beamsplitter’s transmittance, representing the fraction of the signal delivered to Bob. The remaining fraction, \(1 - \eta\), is intercepted by Willie.

The environmental mode $\hat{e}$ is assumed to be in a zero-mean thermal state $\hat{\rho}_{\bar{n}_B}$ with mean photon number $\bar{n}_B$. In the Fock basis, this state is given by the diagonal density operator:
\[
\hat{\rho}_{\bar{n}_B} = \sum_{k=0}^{\infty} \frac{(\bar{n}_B)^k}{(1+\bar{n}_B)^{k+1}} \ket{k}\bra{k}.
\]

This channel model captures the competing roles of transmittance and noise in determining both Bob's ability to decode the message and Willie's ability to detect transmission. In our numerical evaluations, these parameters \((\eta,\bar n_B)\) are instantiated using reported measurements from urban nighttime FSO and short-reach fiber links (see Section~\ref{sec:discussion_implications}).
\vspace{-3mm}

\subsection{Covertness Condition and Security Metrics}
From Willie's perspective as the adversary, the objective is to determine whether Alice is transmitting or not. His optimal detection strategy is formally modeled as a binary quantum hypothesis test between two scenarios:

\begin{itemize}    
    \item \textbf{$H_0$}: Alice is silent for the entire frame. Willie observes the joint product state over $n$ time slots, which is expressed as $\bigotimes_{t=1}^{n} \hat{\rho}_{0_t}^W$. Here, $\hat{\rho}_{0_t}^W$ represents the noise state Willie observes in time slot $t$. Since the noise state is identical for all slots, the joint state for the entire frame can be written more compactly as $(\hat{\rho}_0^W)^{\otimes n}$.

    \item \textbf{$H_1$}: Alice transmits according to the probabilistic scheme described above. Willie observes the perturbed state $\hat{\rho}_1^W$.
\end{itemize}

To achieve covertness, Alice must ensure that Willie's ability to distinguish between these two hypotheses is limited. Willie's detection task involves two types of errors: a \textit{false alarm} (a.k.a. type I error), where he decides $H_1$ when $H_0$ is correct, and a \textit{missed detection} (a.k.a. type II error), where he decides $H_0$ when $H_1$ is correct. These correspond to the error probabilities $P_{\text{FA}} = P(\text{decides } H_1 \mid H_0)$ and $P_{\text{MD}} = P(\text{decides } H_0 \mid H_1)$, respectively. Following the standard convention in covert communication, we assume equal prior probabilities, $P(H_0) = P(H_1) = 1/2$, which leads to an average error for Willie $P_e = \tfrac{1}{2}(P_{\text{FA}} + P_{\text{MD}})$.  The formal covertness condition requires that Willie's error probability, $P_{e}$, be close to that of random guessing~\cite{arghavani2023covert}:
\begin{equation}\label{eq:covertness_1}
P_{e} \ge \frac{1}{2} - \delta,
\end{equation}
for some small security parameter $\delta > 0$. To operationalize the covertness condition in Eq.~\eqref{eq:covertness_1}, we must relate it to the quantum states Willie observes. The Helstrom bound~\cite{helstrom1976quantum} provides this link, connecting the minimum error probability to the trace distance between the two hypotheses, namely $\|\hat{\rho}_1^W - (\hat{\rho}_0^W)^{\otimes n}\|_1$. The minimum error probability Willie can achieve is given by the Helstrom bound:
\[
    P_{e, \min} = \frac{1}{2} \left(1 - \frac{1}{2} \left\| \hat{\rho}_1^W - (\hat{\rho}_0^W)^{\otimes n} \right\|_1 \right).
\]

Substituting the Helstrom bound into Eq.~\eqref{eq:covertness_1} yields:
\[
    \frac{1}{2} \left(1 - \frac{1}{2} \left\| \hat{\rho}_1^W - (\hat{\rho}_0^W)^{\otimes n} \right\|_1 \right) \ge \frac{1}{2} - \delta.
\]
Rearranging this inequality leads to the equivalent condition for the trace distance:
\begin{equation}
    \label{eq:trace_dist_bound}
    \|\hat{\rho}_1^W - (\hat{\rho}_0^W)^{\otimes n} \|_1 \le 4\delta.
\end{equation}

While the trace distance provides a fundamental security guarantee, a more tractable approach involves using the quantum Pinsker inequality~\cite{anderson2024square, wilde2017quantum} to derive a sufficient condition based on the Quantum Relative Entropy (QRE), i.e., $D(\cdot \| \cdot)$.
The Pinsker inequality relates the trace distance ($T = \frac{1}{2}\|\hat{\rho}_1^W-(\hat{\rho}_0^W)^{\otimes n}\|_1$) and the QRE via the bound
\[
    T \le \sqrt{\frac{1}{2} D(\hat{\rho}_1^W \| (\hat{\rho}_0^W)^{\otimes n})}.
\]
To ensure that the covertness requirement $T \le 2\delta$, as given in Eq.~(\ref{eq:trace_dist_bound}), is satisfied, we can impose a stricter condition on the QRE. Integrating the Pinsker inequality with the covertness requirement provides a modified covertness condition:
\[
    \sqrt{\frac{1}{2} D(\hat{\rho}_1^W \| (\hat{\rho}_0^W)^{\otimes n})} \le 2\delta.
\]
Squaring both sides and rearranging yields our final, analytically convenient covertness requirement
\begin{equation}
\label{eq:qre_bound} 
    D(\hat{\rho}_1^W \| (\hat{\rho}_0^W)^{\otimes n}) \le 8\delta^2.
\end{equation}

At this point, we emphasize that the covertness derivation below follows the small-$q$ asymptotic framework of \cite{anderson2024covert, anderson2024square}. In particular, the additivity reduction and the $\chi^2$-based expansion are used in the covert regime, where the randomized transmission probability scales as $q=\mathcal{O}(n^{-1/2})$. Accordingly, the resulting expression for $c_{\mathrm{cov}}$ should be interpreted as the leading-order constant governing the square-root-law covertness constraint in this paper.

To obtain an analytically tractable covertness bound, we employ the quantum $\chi^2$-divergence as an upper bound on the QRE within this small-$q$ asymptotic regime, ensuring that the resulting design rule remains conservative at the level of the adopted approximation. Following the approach in~\cite{anderson2024square}, our analysis relies on the key assumption that the channel to Willie is entanglement-breaking. An entanglement-breaking channel destroys any initial entanglement between a transmitted quantum state and an external system~\cite{anderson2024covert,anderson2024square}. This assumption is crucial for analytical tractability because it ensures that the QRE is additive across multiple channel uses, which simplifies the covertness analysis significantly:\footnote{ While common in this context, the necessity of an entanglement-breaking channel for covert communication is still an open area of research.}
\[
D(\hat{\rho}_1^W \| (\hat{\rho}_0^W)^{\otimes n}) \approx n \cdot D(\hat{\rho}_{1,\text{slot}}^W \| \hat{\rho}_0^W).
\]

For small $q$, the per-slot state \( \hat{\rho}_{1,\text{slot}}^W \) is only a slight perturbation of \( \hat{\rho}_0^W \), allowing a bound via the quantum $\chi^2$-divergence, denoted by $D_{\chi^2}(\cdot\|\cdot)$, which provides an analytically tractable upper bound on the quantum relative entropy:
\[
D(\hat{\rho}_{1,\text{slot}}^W \| \hat{\rho}_0^W) \le D_{\chi^2}(\hat{\rho}_{1,\text{slot}}^W \| \hat{\rho}_0^W) \approx q^2 \cdot \frac{(1-\eta)^2}{\eta \overline{n}_B (1 + \eta \overline{n}_B)}.
\]

Substituting into the QRE bound in Eq.~\eqref{eq:qre_bound} gives:
\[
n \cdot q^2 \cdot \frac{(1 - \eta)^2}{\eta \overline{n}_B (1 + \eta \overline{n}_B)} \lesssim 8\delta^2.
\]
Rearranging this yields a square-root law constraint on \( q \), where the channel-dependent terms are consolidated into the \emph{covertness constant}:
\begin{equation}
\label{eq:ccov}
c_{\text{cov}} = \frac{\sqrt{2 \eta \, \overline{n}_B (1 + \eta \, \overline{n}_B)}}{1 - \eta}.
\end{equation}
Thus, the covert transmission probability must satisfy~\cite{anderson2024square}:
\[
q \le \frac{2 \delta \cdot c_{\text{cov}}}{\sqrt{n}}.
\]
A higher $c_{\mathrm{cov}}$ enables a larger $q$, which yields a larger covert payload while maintaining covertness. Although constant factors depend on the specific bounding and approximation steps used, the square-root-law scaling remains unchanged.
\vspace{-3mm}
\subsection{The Reliability Condition}
For slots where Alice transmits, Bob must decode a noisy, infinite-dimensional bosonic state rather than a clean qubit. To analyze reliability, Anderson et al.~\cite{anderson2024square} use an effective model capturing noise. Two main error sources are:
\begin{enumerate}
    \item \textbf{Channel noise:} Photon loss (linked to $\eta$) and thermal noise ($\overline{n}_B$) distort the transmitted state.
    \item \textbf{Projection error:} Bob projects the received state onto the dual-rail subspace $\text{span}\{ \ket{01}, \ket{10} \}$, discarding components outside this two-dimensional space. This projection is probabilistic and can fail.
\end{enumerate}

The analysis in~\cite{anderson2024square} models the end-to-end effect on Bob's qubit as an effective depolarizing channel. This channel maps an input qubit state $\rho_{\text{in}}$ to an output state $\rho_{\text{out}}$ according to:
\[
    \rho_{\text{out}} = (1 - p) \rho_{\text{in}} + p \cdot \frac{I}{2}.
\]
In this model, $p$ is the depolarizing parameter, and the term $I/2$ represents the maximally mixed state, which corresponds to complete random noise. Conceptually, this channel describes a process of depolarization: with probability $p$, the qubit's state is erased and replaced by random noise, while with probability $1-p$, it is transmitted without error. The physical channel conditions determine the value of $p$, and its closed-form expression is~\cite{anderson2024square}:
\begin{equation}
    \label{eq:p}
    p = 1 - \frac{\eta}{\left(1 + (1 - \eta)\overline{n}_B \right)^4}.
\end{equation}

This depolarizing model implies an equivalent set of probabilistic Pauli errors. The corresponding Pauli error vector, $\vec{p}$, is given by~\cite{anderson2024square}:
\[
    \vec{p} = \left[ 1 - \frac{3p}{4}, \frac{p}{4}, \frac{p}{4}, \frac{p}{4} \right],
\]
representing the respective probabilities for the identity operation and the three Pauli error types (X, Y, Z).\footnote{The Pauli operators, $\{I, X, Y, Z\}$, form a basis for single-qubit errors. The identity operator ($I$) represents no error, the $X$ operator causes a bit-flip, the $Z$ operator causes a phase-flip, and the $Y$ operator causes both.}

The communication rate achievable over this effective channel, denoted $R$, is given by the quantum \emph{hashing bound}~\cite{anderson2024square}:
\begin{equation}\label{eq:comm_rate}
R = \left[ 1 - H(\vec{p}) \right]^+,
\end{equation}
where $H(\vec{p}) = -\sum_i p_i \log_2 p_i$ is the Shannon entropy and $[x]^+ = \max(x, 0)$. This rate quantifies reliability, the maximum throughput per transmitted qubit. It trades off directly with the transmit probability $q$ set by covertness. Together, these constraints define the security envelope of the system.
\vspace{-3mm}
\subsection{Transition to Robustness and Uncertainty}
\label{sec:uncertainty_transition}

While the preceding metrics characterize performance for a
fixed $(\eta,\bar n_B)$, we now extend the analysis to a passive
warden setting with bounded channel uncertainty. To ensure high-assurance operation under environmental variability and estimation error, we model the unknown parameters through the bounded uncertainty set $\mathcal{U} = [\eta_{\min}, \eta_{\max}] \times [\bar{n}_{B,\min}, \bar{n}_{B,\max}]$. Here, Nature is a modeling abstraction that selects a realization from $\mathcal{U}$ adversarially so as to minimize Alice's guaranteed performance. The core analytical challenge is that covertness and reliability are generally extremized by different corners of $\mathcal{U}$. A robust design must therefore satisfy both constraints simultaneously across conflicting realizations to guarantee a lower bound on the covert payload.

\section{Robust Performance Bounds Under Uncertainty}
\label{sec:robust}

The idealized model discussed thus far assumes that Alice has perfect, instantaneous knowledge of the channel parameters. We now depart from this assumption to consider a more practical scenario where Alice has only bounded, deterministic knowledge of the transmittance and noise level:
$\eta \in [\eta_{\min}, \eta_{\max}]$, and $\bar n_B \in [\bar n_{B,\min}, \bar n_{B,\max}]$. This bounded uncertainty model marks an essential step toward a more realistic analysis. The assumption of perfect channel knowledge is physically unattainable, as all measurements are subject to noise and real-world channels inherently fluctuate. The uncertainty bounds $[\eta_{\min}, \eta_{\max}]$ and $[\bar n_{B,\min}, \bar n_{B,\max}]$ are not arbitrary. In practice, they would be estimated using a combination of active channel probing, environmental sensing, and predictive modeling before covert transmission. For instance, Alice could send faint, non-covert pilot signals to measure transmittance fluctuations over time, use radiometers to sense the ambient thermal background, or employ atmospheric models to predict the range of expected parameters.% for a given link.

However, these estimation procedures are themselves subject to statistical error, which may lead to a mismatch between the assumed and actual uncertainty regions. If the true channel parameters drift outside the estimated bounds, the security guarantees derived from this model would no longer be valid. A fully robust design must therefore consider not only the physical channel variation but also uncertainty in the estimated bounds themselves, a higher-order risk that motivates future work on adaptive bound calibration. Our framework thus addresses the fundamental engineering question: how can one design a covert communication system that guarantees performance despite imperfect knowledge of its environment?

The central challenge is to design a transmission strategy that guarantees both covertness and reliability for \emph{all} possible realizations of \( \eta \) and \( \bar n_B \) within these bounds. To derive such robust guarantees, Alice must adopt a conservative strategy based on a worst-case analysis. This approach, in turn, requires first establishing the monotonic behavior of the performance metrics with respect to the channel parameters, which we accomplish in the following.

\subsubsection*{Monotonicity and corner extremizers}
We first establish monotonicity of the key quantities with respect to $(\eta,\bar n_B)$; these immediately imply which corners of the uncertainty box govern robust design.

\begin{lemma} \label{lemma:ccov}
The covertness constant, $c_{\text{cov}}$ (Eq. \eqref{eq:ccov}), is a monotonically increasing function of both the transmittance $\eta$ and the mean thermal photon number $\bar n_B$ for all parameters $\eta \in (0,1)$ and $\bar n_B > 0$.
\end{lemma}

\begin{proof}
    See Appendix~\ref{app:proofs-ccov}.
\end{proof}

\noindent\emph{Takeaway.} Higher $\eta$ and higher $\bar n_B$ both make covertness easier (larger $c_{\rm cov}$).

\begin{corollary}[Robust Covertness Constant]
\label{cor:robust_covertness}
To guarantee covertness for all channel parameters in the uncertainty set $\mathcal{U}$, where $\eta \in [\eta_{\min}, \eta_{\max}]$ and $\bar n_B \in [\bar n_{B,\min}, \bar n_{B,\max}]$, the transmission strategy must be based on the robust covertness constant $c_{\mathrm{cov}}^{\mathrm{robust}}$, defined as the minimum value of $c_{\mathrm{cov}}(\eta, \bar n_B)$:
\begin{equation}
\begin{split}    
    c_{\mathrm{cov}}^{\mathrm{robust}} = c_{\mathrm{cov}}(\eta_{\min}, \bar n_{B,\min}) =\\
    \frac{\sqrt{2\eta_{\min}\bar n_{B,\min}(1+\eta_{\min}\bar n_{B,\min})}}{1-\eta_{\min}}.    
\end{split}
\end{equation}
\end{corollary}

\begin{proof}
Since $c_{\mathrm{cov}}$ increases in both arguments (Lemma~\ref{lemma:ccov}),
its minimum over
$[\eta_{\min},\eta_{\max}]\times[\bar n_{B,\min},\bar n_{B,\max}]$
occurs at $(\eta_{\min},\bar n_{B,\min})$.
\end{proof}

\noindent\emph{Takeaway.} The robust design is governed by Willie’s most favorable channel realization, namely $(\eta_{\min},\bar n_{B,\min})$.

\begin{lemma} \label{lemma:p}
The depolarizing parameter, $p$, is a monotonically increasing function of the mean thermal photon number $\bar n_B$ and a monotonically decreasing function of the transmittance $\eta$ for all valid parameters.
\end{lemma}

\begin{proof}
See Appendix~\ref{app:proofs-p}.
\end{proof}

\noindent\emph{Takeaway.} Reliability degrades with higher noise and lower transmittance.

\begin{corollary}[Worst-Case Depolarizing Parameter]
\label{cor:worst_case_p}
To guarantee a reliable communication rate $R$ for all channel parameters within the uncertainty region $\eta \in [\eta_{\min}, \eta_{\max}]$ and $\bar n_B \in [\bar n_{B,\min}, \bar n_{B,\max}]$, the rate must be calculated from the worst-case depolarizing parameter, $p_{\mathrm{worst}}$. This parameter corresponds to the maximum value of $p(\eta, \bar n_B)$ over the uncertainty region and is given by:
\begin{equation}
    p_{\mathrm{worst}} = p(\eta_{\min}, \bar n_{B,\max}) = 1 - \frac{\eta_{\min}}{\left(1 + (1 - \eta_{\min})\bar n_{B,\max}\right)^4}.
\end{equation}
\end{corollary}

\begin{proof}
By Lemma~\ref{lemma:p}, $p$ increases in $\bar n_B$ and decreases in $\eta$, so its maximum over the uncertainty set occurs at $(\eta_{\min},\bar n_{B,\max})$.
\end{proof}

\noindent\emph{Takeaway.} Compute $R$ using $(\eta_{\min},\bar n_{B,\max})$.

\begin{theorem}[Robust Policy and Guaranteed Payload Bound]\label{thm:main_result}
Let
\[
\mathcal{U} = [\eta_{\min},\eta_{\max}] \times [\bar n_{B,\min},\bar n_{B,\max}]
\]
denote the compound uncertainty set. Define the robust transmission probability and robust coding rate by
\[
q_{\mathrm{rob}} \triangleq \frac{2\delta\, c_{\mathrm{cov}}^{\mathrm{robust}}}{\sqrt{n}},
\qquad
R_{\mathrm{rob}} \triangleq R_{\mathrm{worst}},
\]
where $c_{\mathrm{cov}}^{\mathrm{robust}} = c_{\mathrm{cov}}(\eta_{\min},\bar n_{B,\min})$, $R_{\mathrm{worst}} = \bigl[1-H(\vec p_{\mathrm{worst}})\bigr]^+$, and
$p_{\mathrm{worst}} = p(\eta_{\min},\bar n_{B,\max})$.

Then the policy $(q_{\mathrm{rob}},R_{\mathrm{rob}})$ satisfies the covertness and reliability constraints for every $(\eta,\bar n_B)\in\mathcal{U}$. Consequently, the guaranteed worst-case expected covert payload satisfies
\begin{equation}\label{eq:main_result}
\begin{aligned}
M_{\mathrm{rob}}(n)
&\triangleq \inf_{(\eta,\bar n_B)\in\mathcal{U}}
\mathbb{E}[M(n)\mid q_{\mathrm{rob}},R_{\mathrm{rob}}]
\\
&\ge
2\sqrt{n}\, c_{\mathrm{cov}}^{\mathrm{robust}}\, R_{\mathrm{worst}}\, \delta.
\end{aligned}
\end{equation}
\end{theorem}

\begin{proof}
    See Appendix~\ref{app:Proof_thm1}.
\end{proof}

Within the class of static, channel-independent $(q,R)$ policies that satisfy the covertness and reliability constraints uniformly over $\mathcal{U}$, the policy $(q_{\mathrm{rob}},R_{\mathrm{rob}})$ is optimal for maximizing the guaranteed worst-case expected covert payload. Indeed, any such feasible policy must satisfy $q\le q_{\mathrm{rob}}$ and $R\le R_{\mathrm{worst}}$, while for a fixed static policy the guaranteed payload is $\inf_{(\eta,\bar n_B)\in\mathcal{U}}\mathbb{E}[M(n)\mid q,R]=nqR$, which is increasing in both variables.

In other words, Eq.~\eqref{eq:main_result} gives a guaranteed lower bound on the payload attained by the robust policy that saturates the worst-case covertness constraint and uses the worst-case reliably achievable rate. In practice, the relative errors on transmittance and thermal background need not be symmetric around $(\eta_0,\bar n_{B,0})$. We therefore generalize the rectangular uncertainty set to \emph{asymmetric} boxes
\begin{equation}
\begin{aligned}
\eta \in \big[(1-a)\eta_0,\ \min\{(1+b)\eta_0,1\}\big], \\
\bar n_B \in \big[(1-c)\bar n_{B,0},\ (1+d)\bar n_{B,0}\big].
\end{aligned}
\label{eq:asym_boxes}
\end{equation}
with $a,b,c,d \in [0,1)$ describing relative under/overestimation margins, and write
\[
\eta_{\min}=(1-a)\eta_0,\;\; \eta_{\max}=\min\{(1+b)\eta_0,1\},
\]
\[
\bar n_{B,\min}=(1-c)\bar n_{B,0},\;\; \bar n_{B,\max}=(1+d)\bar n_{B,0}.
\]

\begin{remark}[Asymmetric uncertainty boxes]
\label{rem:asym_corners}
For asymmetric margins $\eta \in \big[(1-a)\eta_0,\ \min\{(1+b)\eta_0,1\}\big]$ and 
$\bar n_B \in [(1-c)\bar n_{B,0},(1+d)\bar n_{B,0}]$ with $a,b,c,d\in[0,1)$,
the corner-extremizer conclusions are unchanged:
$c_{\rm cov}$ is minimized at $(\eta_{\min},\bar n_{B,\min})$ and
$p$ is maximized at $(\eta_{\min},\bar n_{B,\max})$. Therefore, Theorem~\ref{thm:main_result} holds unchanged with these endpoints.
\emph{This follows immediately from Lemmas~\ref{lemma:ccov} and~\ref{lemma:p}.}

\emph{Why this matters:} Channel errors are rarely symmetric. The $(a,b,c,d)$ model captures practical skew (e.g., upward noise drift). Since the corner-extremizer property still holds, the design rules remain unchanged (see Fig.~\ref{fig:asym-vs-sym}).
\end{remark}

\section{Analytical and Numerical Results}
\label{sec:discussion}
\subsection{Nominal vs.\ Robust Design}
When designing a covert quantum communication system, Alice must satisfy both covertness and reliability, but these constraints are governed by different channel conditions and impose conflicting design requirements. Covertness is dictated by the Alice-to-Willie channel and is hardest to guarantee under conditions most favorable to Willie’s detection, namely, low environmental noise and low transmittance. Reliability depends on the Alice-to-Bob channel and is hardest to guarantee under poor reception conditions, namely, low transmittance and high noise. Hence, the worst cases for covertness and reliability occur at different corners of the uncertainty region: both are governed by $\eta_{\min}$, but covertness is paired with $\bar n_{B,\min}$ whereas reliability is paired with $\bar n_{B,\max}$. At first glance, this appears contradictory: How can a single transmission strategy handle both low- and high-noise environments? The resolution lies in the separation of constraints: each is evaluated on a distinct path, i.e., covertness on the Alice-to-Willie channel and reliability on the Alice-to-Bob channel. Our robust framework synthesizes these worst cases into a single conservative strategy that simultaneously guarantees both security and performance. This is not a contradiction, but a trade-off between security and reliability, and our performance bound quantifies the cost of reconciling these competing requirements.

To demonstrate the need for robustness, we contrast our design with a baseline that tunes $(q,R)$ only at the nominal channel $(\eta_0,\bar n_{B,0})$ and ignores uncertainty. Although this baseline may show high payload at the nominal point, it does
not offer a nonzero guarantee on the uncertainty set $\mathcal{U}$ defined below. The proposition formalizes this failure mode.

\begin{proposition}[Guaranteed payload collapse of the nominal baseline]\label{prop:baseline_fails}
Fix nominal parameters $0<\eta_0<1$ and $\bar n_{B,0}>0$. Let $0<u<1$ define the uncertainty level, and let $\mathcal{U}$ be the uncertainty set as in Theorem~1. Define
\begin{align}
\mathcal{U}_{\mathrm{base}}
&\triangleq
\{(\eta,\bar n_B): \eta\in[\eta_{\min},\eta_0],\;
\bar n_B\in[\bar n_{B,\min},\bar n_{B,\max}]\}
\nonumber\\
&\subseteq \mathcal{U}.
\end{align}
Consider a naive policy that selects its transmission probability $q_{\mathrm{nom}}$ and coding rate $R_{\mathrm{nom}}$ by saturating the covertness and reliability constraints at the nominal point $(\eta_0,\bar n_{B,0})$. Then this policy is not uniformly feasible over $\mathcal{U}$. Consequently, its guaranteed covert-and-reliable payload over $\mathcal{U}$ is zero.
\begin{equation}
M_{\mathrm{rob,naive}}(n) = 0.
%\tag{13}
\end{equation}
Here, $M_{\mathrm{rob,naive}}(n)$ denotes the certified worst-case covert-and-reliable payload achieved by the nominally tuned policy, with the convention that this quantity is zero whenever the policy violates either the covertness or reliability requirement for some realization in $\mathcal{U}$.
\end{proposition}

\begin{proof}
    See Appendix~\ref{app:Proofs-baseline_fails}
\end{proof}

Fig.~\ref{fig:baseline_vs_robust} illustrates the behavior predicted by Proposition~\ref{prop:baseline_fails}.
Throughout the figure, we fix the relative uncertainty level to $5\%$ (i.e., $u=0.05$) and use the nominal pair $(\eta_0,\bar n_{B,0})$ in Table~\ref{tab:params}.
By \emph{naive (scheduled)}, we mean the payload $\mathbb{E}[M(n)]$ that the nominally tuned baseline plans to transmit when the channel equals $(\eta_0,\bar n_{B,0})$.
By \emph{robust (guaranteed)}, we mean the worst-case certified payload $M_{\mathrm{rob}}(n)$ across the uncertainty set $\mathcal{U}$, calculated from Eq.~\eqref{eq:main_result} using the covertness corner $(\eta_{\min},\bar n_{B,\min})$ and the reliability corner $(\eta_{\min},\bar n_{B,\max})$.
The naive (guaranteed) payload is the baseline’s worst-case value over $\mathcal{U}$ and, by Proposition~\ref{prop:baseline_fails}, is identically zero. 
The figure plots these quantities versus block length $n$: the naive curve (scheduled) reflects the optimistic nominal payload, while the robust curve (guaranteed) reports the provable worst-case certified payload $M_{\mathrm{rob}}(n)$ under uncertainty (the naive–guaranteed curve is omitted because it is zero).
For example, at $n = 10^8$, the nominal baseline would
schedule 1673.9 qubits, but its guaranteed payload under $\mathcal{U}$ is $0$. In contrast, our robust design guarantees 440.2 covert qubits.

\emph{Note.} The larger naive (scheduled) value reflects nominal tuning at $(\eta_0,\bar n_{B,0})$ and is not a guarantee over $\mathcal{U}$; the relevant comparison under uncertainty is between the \emph{guaranteed worst-case} payloads, for which the naive baseline is $0$ whereas the robust design $M_{\mathrm{rob}}(n)$ remains strictly positive.

\begin{figure}[t]
  \centering
\includegraphics[width=0.375\textwidth]{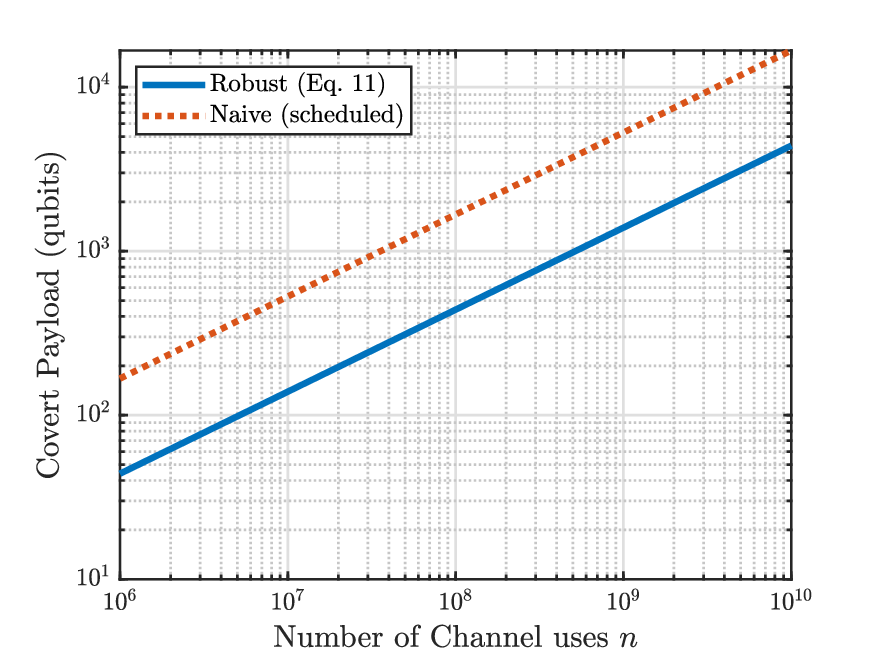}%
%\Description{Subplot (a): scheduled vs. guaranteed payloadacross block lengths under 5\% uncertainty.}%
\caption{\textbf{Naive vs.\ robust under $u=5\%$.} Nominal tuning yields zero guaranteed worst-case covert payload over $\mathcal{U}$, whereas the robust policy in \eqref{eq:main_result} yields a strictly positive guaranteed lower bound $M_{\mathrm{rob}}(n)$. The curves show the nominal covert payload and the robust guaranteed covert payload versus block length $n$. The ``Naive (guaranteed)'' curve is identically zero and is therefore omitted.}
\vspace{-5mm}
\label{fig:baseline_vs_robust}
\end{figure}
\vspace{-3mm}
\subsection{Numerical Framework and Validation}
\label{sec:numerical_framework}
To validate the covertness component of our analytical framework, we developed a QuTiP simulation of the underlying four-mode physical system using nominal parameters consistent with~\cite{anderson2024square} and summarized in Table~\ref{tab:params}. Since each bosonic mode has an infinite-dimensional Fock space, we performed a convergence analysis to determine an accurate truncation cutoff by computing the quantum $\chi^2$-divergence across dimensions (Fig.~\ref{fig:convergence_plot}, Table~\ref{tab:convergence_chi2}). The results show that a Fock basis of dimension 7 yields sufficient accuracy (relative error $<0.01\%$), ensuring higher photon-number contributions are negligible. Using this validated setup, we numerically estimate the Willie-side quantum $\chi^2$-divergence, infer the corresponding covertness constant, and combine it with the analytical reliability rate to obtain the guaranteed robust lower bound $M_{\mathrm{rob}}(n)$.

\begin{table}[!t]
\centering
\caption{\textbf{Simulation parameters.} $\eta_0=0.9$ and $\bar n_{B,0}=0.12$ reflect realistic short-haul fiber or nighttime free-space optical links with dark counts and ambient noise~\cite{bourgoin2013comprehensive}.}
\label{tab:params}
\begin{tabular}{l|c|l}
\hline
%\midrule
\textbf{Parameter} & \textbf{Symbol} & \textbf{Value} \\ \hline
Nominal Transmittance & $\eta_0$ & 0.9 \\
Nominal Mean Thermal Photon No. & $\bar n_{B,0}$ & 0.12 \\
Covertness Parameter & $\delta$ & 0.05 \\
Number of Channel Uses (variable)& $n$ & $10^6$ -- $10^{10}$ \\
\hline
\end{tabular}
\vspace{-3mm}
\end{table}

%--- Figure and Table for Convergence ---%
\begin{figure}[!t]
    \centering
    \includegraphics[width=0.375\textwidth]{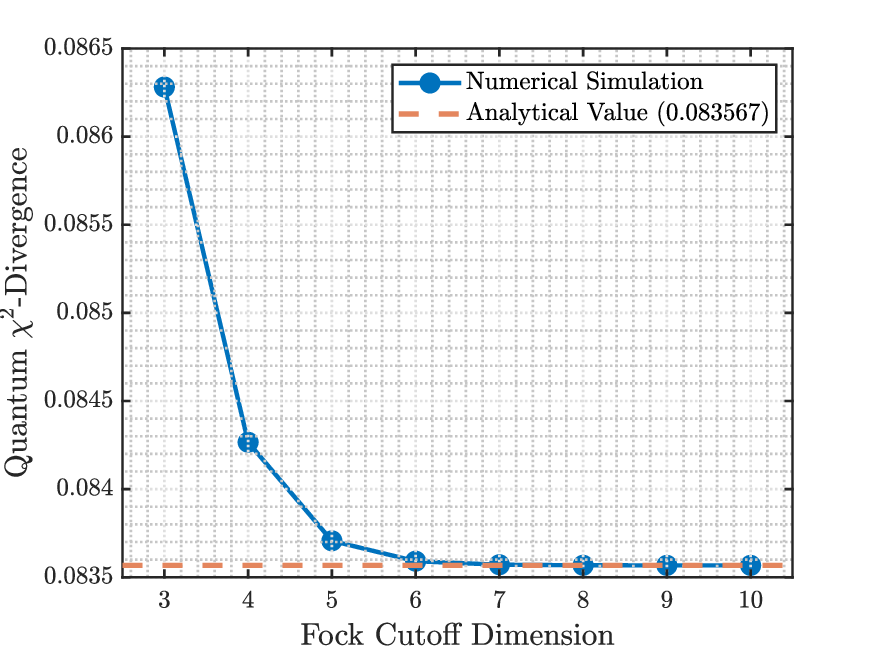}
    %\Description{Plot showing the simulated quantum chi-squared divergence versus Fock space cutoff dimension. Blue markers represent QuTiP simulation results and the red dashed line the analytical reference, demonstrating rapid convergence at cutoff dimension seven.}
    \caption{\textbf{Convergence of simulated $D_{\chi^2}$ vs.\ Fock cutoff.} $\eta=0.9,\ \bar n_B=0.12$. QuTiP markers approach the dashed analytical value, validating \texttt{cutoff\_dim}=7.
    %\textbf{Convergence of the simulated quantum $\chi^2$-divergence as a function of the Fock space cutoff dimension.} The simulation was performed with nominal channel parameters $\eta = 0.9$ and $\bar{n}_B = 0.12$. Blue markers indicate the numerical values obtained using QuTiP, while the red dashed line shows the exact analytical reference. The plot demonstrates rapid convergence and validates the use of \texttt{cutoff\_dim = 7}.}% for the main results.
    }
    \vspace{-5mm}
    \label{fig:convergence_plot}
\end{figure}

\begin{table}[!ht]
\centering
\caption{Convergence of simulated $\chi^2$-divergence vs. Fock cutoff dimension ($\eta=0.9,\ \bar n_B=0.12$). Analytical value: $0.08356732$.}
\label{tab:convergence_chi2}
% Set \tabcolsep to 0pt so tabular* has full control over spacing.
\setlength{\tabcolsep}{0pt} 
\begin{tabular*}{0.8\columnwidth}{@{\extracolsep{\fill}}cccc}
\hline
\textbf{Cutoff Dim.} & \textbf{Simulated $\boldsymbol{\chi^2}$} & \textbf{Absolute Error} & \textbf{Relative Error (\%)} \\
\hline
\hline
3  & 0.08628140 & 0.00271408 & 3.2478 \\
4  & 0.08426596 & 0.00069864 & 0.8360 \\
5  & 0.08370664 & 0.00013932 & 0.1667 \\
6  & 0.08359112 & 0.00002379 & 0.0285 \\
7  & 0.08357102 & 0.00000370 & 0.0044 \\
8  & 0.08356786 & 0.00000054 & 0.0006 \\
9  & 0.08356740 & 0.00000008 & 0.0001 \\
10 & 0.08356733 & 0.00000001 & 0.0000 \\
\hline
\end{tabular*}
\end{table}

As shown in Fig.~\ref{fig:convergence_plot} and Table~\ref{tab:convergence_chi2}, the numerical simulation matches the analytical prediction with high precision for the tested operating point. These results validate the numerical framework used in the subsequent performance evaluations.

\subsection{The Overall Impact of Channel Uncertainty}
We first investigate the impact of uncertainty on the total
covert payload. Under perfect channel knowledge, this payload is given by $\mathbb{E}[M(n)]$, whereas under bounded uncertainty the corresponding guaranteed payload is $M_{\mathrm{rob}}(n)$ from Eq.~\ref{eq:main_result}. Fig.~\ref{fig:cost_of_uncertainty} plots the expected covert payload, $\mathbb{E}[M(n)]$, as a function of the number of channel uses, $n$. The solid blue line represents the ideal scenario where Alice has perfect knowledge of the channel parameters. The dashed lines show the robust performance bound derived in Eq.~\eqref{eq:main_result} for several relative uncertainty levels (where a symmetric uncertainty level $u\in(0,1)$ means that both parameters vary within $\pm 100u\%$ of their nominal values, subject to physical admissibility, i.e., $\eta \in \bigl[(1-u)\eta_0,\ \min\{(1+u)\eta_0,1\}\bigr]$ and $\bar n_B \in \bigl[(1-u)\bar n_{B,0},\ (1+u)\bar n_{B,0}\bigr]$). As expected, there is a significant performance gap between the ideal and robust cases, which widens as uncertainty increases. This gap represents the ``cost of uncertainty'' that Alice must
pay to provide a robust security and reliability guarantee.

\begin{figure}[!t]
    \centering
        \includegraphics[width=0.375\textwidth]{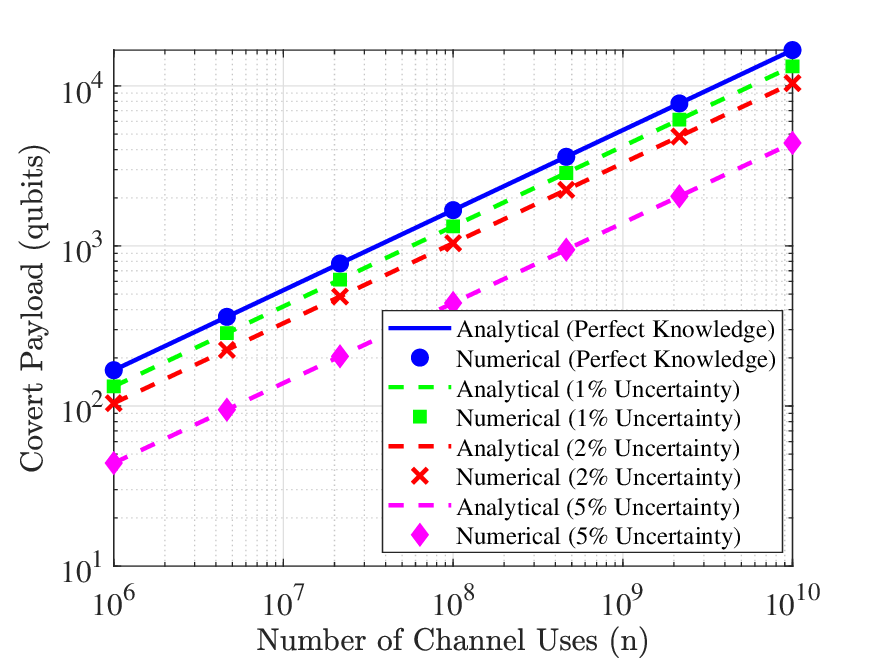}
        %\Description{Line plot showing the relationship between the number of channel uses and the guaranteed number of secure qubits. The figure compares the ideal case with perfect channel knowledge to robust security-performance bounds under 1\%, 2.5\%, and 5\% channel parameter uncertainty, illustrating how greater uncertainty reduces achievable secure qubits.}
        \caption{\textbf{Covert payload vs.\ channel uses.} Comparison between the perfect-knowledge covert payload and the guaranteed robust lower bound for 1\%, 2\%, and 5\% uncertainty. The widening gap illustrates the cost of uncertainty required to maintain worst-case guarantees.}
        \label{fig:cost_of_uncertainty}
        \vspace{-3mm}
\end{figure}
\begin{figure}
        \centering
        \includegraphics[width=0.375\textwidth]{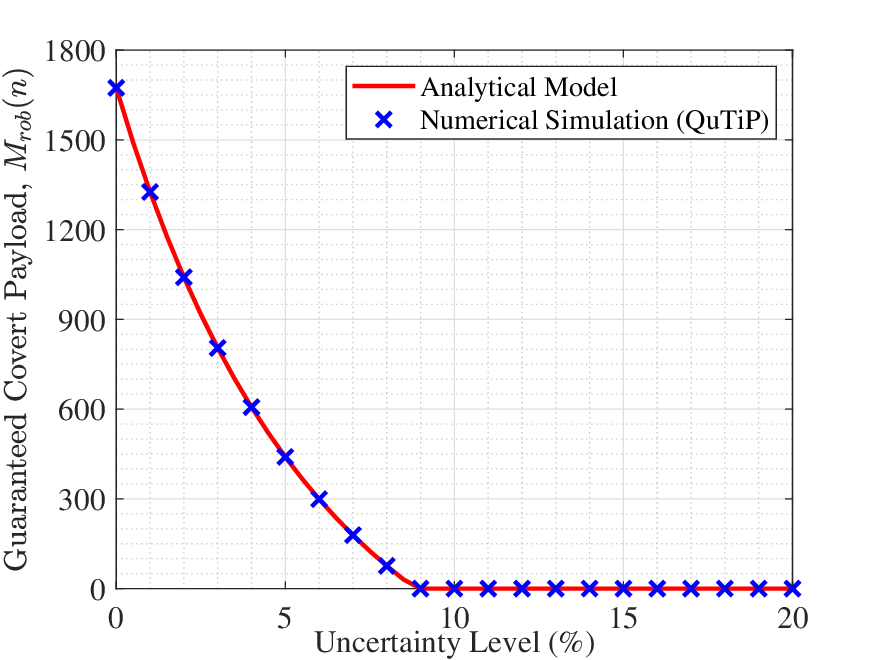}
        %\Description{Plot showing secure throughput versus channel uncertainty for a fixed block length of $n = 10^8$. The curve declines gradually and drops to zero beyond the critical uncertainty threshold of 8.85\%, illustrating a sharp boundary—known as the rate cliff edge—beyond which secure covert communication is no longer achievable.}
        \caption{\textbf{Rate cliff edge.} For $n=10^8$, guaranteed covert payload falls to zero beyond 8.85\% uncertainty, indicating a hard boundary for robust covert communication.
        %\textbf{The Rate Cliff Edge as a Hard Security Boundary.} This plot shows the degradation of the guaranteed secure throughput as channel uncertainty increases for a fixed $n=10^8$. Performance drops to zero beyond the critical uncertainty threshold of 8.85\%, demonstrating a hard limit for maintaining a secure and reliable covert channel.
        }
        \label{fig:performance_vs_uncertainty}
        \vspace{-3mm}
\end{figure}

To further explore this degradation, Fig.~\ref{fig:performance_vs_uncertainty} shows the robust performance for a fixed number of channel uses ($n=10^8$) as the uncertainty level is varied. The plot clearly illustrates that the guaranteed covert payload decreases monotonically as Alice's knowledge of the channel becomes less precise. Notably, the performance does not degrade gracefully to zero. Instead, it encounters a sharp ``cliff edge,'' beyond which the certified worst-case reliable rate under the adopted hashing-bound criterion drops to zero $(R_{\mathrm{worst}}=0)$. For our chosen parameters, this threshold occurs at a critical uncertainty level of approximately 8.85\%. This threshold effect is a fundamental feature of the adopted reliability model: once the worst-case effective channel crosses a certain noise level, the certified achievable rate given by the hashing bound vanishes.
We discuss the origin of this threshold in detail in Section~\ref{subsubsec:cliff_edge}.

\noindent\textbf{Symmetric vs.\ asymmetric uncertainty.}
Beyond the symmetric $u$–sweep, it is informative to compare the robust guarantee against a fixed \emph{asymmetric} box that reflects skewed estimation margins on $(\eta,\overline{n}_B)$.

\begin{figure}[t]
\centering
\includegraphics[width=0.375\textwidth]{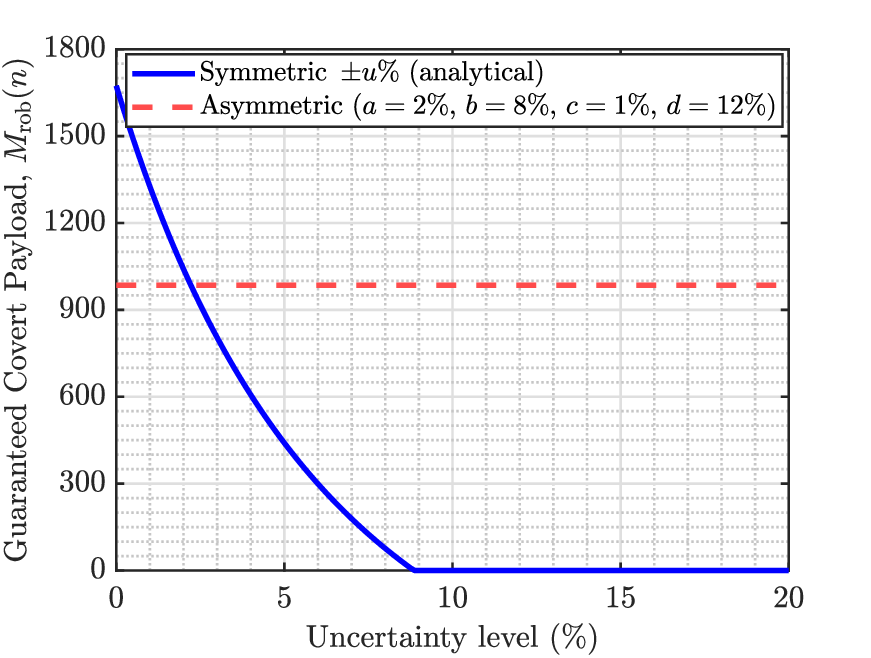}
\caption{\textbf{Symmetric vs.\ asymmetric robust bounds.} Guaranteed covert payload
$M_{\mathrm{rob}}(n)$ under symmetric uncertainty $u$ (solid) vs. an asymmetric box $(a,b,c,d)=(0.02,0.08,0.01,0.12)$ around $(\eta_0,\bar n_{B,0})=(0.9,0.12)$, with $(n,\delta)=(10^8,0.05)$. The symmetric–asymmetric intersection gives the equivalent symmetric margin.
%\textbf{Symmetric vs.\ asymmetric robust bound.} Robust $\mathbb{E}[M(n)]$ under symmetric uncertainty $u$ (solid) and a fixed asymmetric box $(a,b,c,d)=(0.02,0.08,0.01,0.12)$ at $(\eta_0,\overline{n}_{B,0})=(0.90,0.12)$ with $(n,\delta)=(10^8,0.05)$. The asymmetric guarantee is flat; the intersection gives the equivalent symmetric margin.
}
%\Description{Line plot of robust E[M(n)] versus uncertainty u. A solid curve decreases with u; a dashed horizontal line shows the fixed asymmetric guarantee. Axes: x is uncertainty level in percent, y is E[M(n)]. Parameters: eta0=0.90, nB0=0.12, n=1e8, delta=0.05; asymmetry a=2\%, b=8\%, c=1\%, d=12\%.}
\label{fig:asym-vs-sym}
\vspace{-5mm}
\end{figure}

As shown in Fig.~\ref{fig:asym-vs-sym}, the asymmetric guarantee is flat in $u$ once $(a,b,c,d)$ are fixed because the extremizers are corners (Remark~\ref{rem:asym_corners}). The symmetric guarantee $M_{\mathrm{rob}}^{\mathrm{sym}}(u)$ decays with $u$; their crossing provides an ``exchange rate'' between the chosen skew and an equivalent symmetric margin.

\begin{figure}[!b]
\centering
\includegraphics[width=0.375\textwidth]{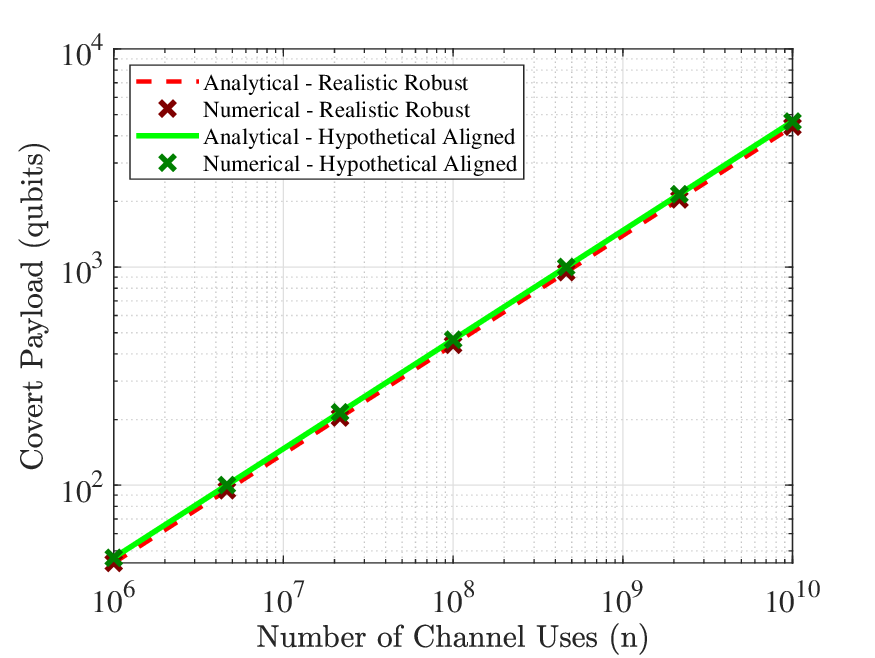}
        \caption{\textbf{Cost of the covertness--reliability conflict.} At $u=5\%$ uncertainty, the realistic robust payload $M_{\mathrm{rob}}(n)$ is compared with a hypothetical aligned bound in which both covertness and reliability are artificially evaluated at the same channel realization, chosen as $(\eta_{\min},\bar n_{B,\max})$. The gap isolates the performance loss caused by the mismatch between the true covertness extremizer $(\eta_{\min},\bar n_{B,\min})$ and the true reliability extremizer $(\eta_{\min},\bar n_{B,\max})$.}
\label{fig:conflict_decomposition}
\end{figure}

\vspace{-3mm}
\subsection{Quantifying the Cost of Constraint Conflict}

Having established that uncertainty degrades performance, we now isolate and quantify the specific penalty that arises from the conflicting nature of the covertness and reliability constraints. Fig.~\ref{fig:conflict_decomposition} decomposes the performance penalty for a 5\% uncertainty level. It compares our realistic robust bound against a hypothetical aligned bound in which both covertness and reliability are artificially evaluated at the same channel realization, chosen here as the reliability corner $(\eta_{\min},\bar n_{B,\max})$. In other words, the aligned comparator removes the mismatch between the two extremizers by forcing both constraints to be evaluated at Bob's worst-case corner, even though the true robust covertness constraint is governed by $(\eta_{\min},\bar n_{B,\min})$. The gap between these two curves reveals a performance loss of 5.32\% that is attributable solely to this extremizer mismatch. This demonstrates that a substantial portion of the overall performance degradation is a direct consequence of the system being pulled in two different directions by the requirements of security and reliability.

\begin{figure}[!t]
\centering
\includegraphics[width=0.375\textwidth]{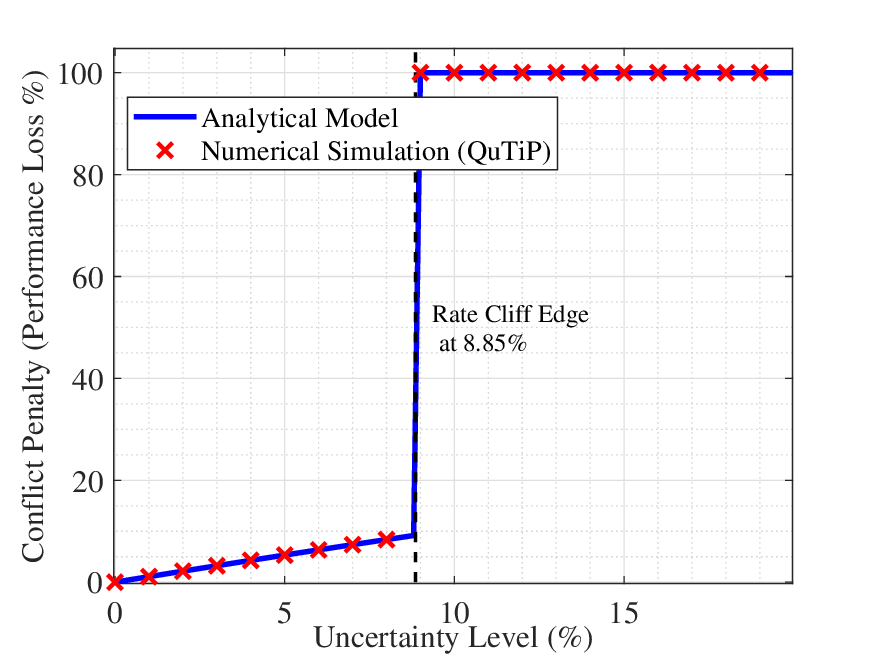}
        %\Description{Plot showing the percentage of performance lost due to the security–reliability conflict versus channel uncertainty. The security tax grows with uncertainty and leads to a complete loss of the secure channel near the rate cliff edge.}
        \caption{\textbf{Security tax vs.\ uncertainty.} Percentage performance loss attributable to the covertness--reliability conflict. The security tax grows with uncertainty up to the rate-cliff edge. Beyond the cliff, the 100\% level is shown as a plotting convention to indicate complete collapse of feasible guaranteed covert communication under the adopted worst-case criterion, rather than as a literal ratio between two strictly positive payloads.}       
        \label{fig:penalty_vs_uncertainty}
        \vspace{-5mm}
\end{figure}

Finally, Fig.~\ref{fig:penalty_vs_uncertainty} quantifies this
security tax as a function of the uncertainty level. The plot shows
that the percentage performance loss attributable to the
covertness--reliability conflict grows as the uncertainty increases.
This loss grows steadily until it reaches the critical uncertainty threshold (the ``rate cliff edge''), at which point the guaranteed covert payload drops to zero. Beyond this threshold, feasible guaranteed covert communication is lost under the adopted worst-case criterion. Accordingly, the post-cliff portion of Fig.~\ref{fig:penalty_vs_uncertainty} is displayed as a plotting convention indicating complete collapse of the robust operating point, rather than as a literal ratio between two strictly positive payloads.
This result illustrates the severity of the conflict: as Alice's
knowledge of the channel worsens, the cost of satisfying two diverging worst-case conditions simultaneously becomes increasingly dominant, ultimately defining the boundary of
feasible covert communication.
\vspace{-4mm}
\subsection{The Rate Cliff Edge Threshold}
\label{subsubsec:cliff_edge}
A notable feature observed in our numerical results is the ``rate cliff edge'', a critical uncertainty level beyond which the guaranteed covert payload abruptly drops to zero. This is not a numerical artifact but a threshold effect inherent in the adopted reliability model. In our framework, reliability is certified through the hashing-bound achievable rate of the effective depolarizing channel. Once the worst-case channel becomes sufficiently noisy, this certified achievable rate vanishes. In our model, the relevant rate is $R_{\text{worst}} = \left[1 - H(\vec{p}_{\text{worst}})\right]^+$, where $H(\cdot)$ is the Shannon entropy. The rate drops to zero if and only if the entropy of the worst-case equivalent Pauli channel meets or exceeds one, i.e., $H(\vec{p}_{\text{worst}}) \ge 1$. This allows us to define a critical depolarizing probability, $p_{\text{crit}}$, as the exact value of $p$ for which the entropy equals one. This constant is found by numerically solving the equation:
\begin{equation}
\begin{split}
    -\left(1 - \frac{3p_{\text{crit}}}{4}\right)\log_2\left(1 - \frac{3p_{\text{crit}}}{4}\right) - 3 \left(\frac{p_{\text{crit}}}{4}\right)\log_2\left(\frac{p_{\text{crit}}}{4}\right) = 1.
\end{split}
\label{eq:hashing_bound_ths}
\end{equation}
The cliff edge observed in our simulations corresponds to the uncertainty level at which the worst-case depolarizing probability for Bob reaches this threshold. This is formally stated in the following corollary.

\begin{corollary}[The Rate Cliff Edge]
\label{cor:cliff_edge}
Let $(\eta_0, \bar{n}_{B,0})$ be the nominal channel parameters and $u$ be the uncertainty level. The guaranteed reliable rate $R_{\text{worst}}$ is zero if and only if $p_{\text{worst}}(u) \ge p_{\text{crit}}$, where $p_{\text{crit}} \approx 0.2524$ is the threshold at which the hashing bound \eqref{eq:comm_rate} vanishes. Whenever the threshold is crossed within the admissible uncertainty range, the critical uncertainty level $u_{\text{crit}}$ is the minimal solution of
\begin{equation}
    p_{\text{crit}} = 1 - \frac{\eta_0(1-u_{\text{crit}})}{\left(1 + [1 - \eta_0(1-u_{\text{crit}})]\bar{n}_{B,0}(1+u_{\text{crit}}) \right)^4}.
    \label{eq:p_crit}
\end{equation}
\end{corollary}

\begin{proof}
By \eqref{eq:comm_rate}, $R_{\text{worst}} = 0$ when $H(\vec{p}_{\text{worst}}) \ge 1$. Since $p_{\text{worst}}(u)$ is monotonically nondecreasing in $u$ (by Lemma~\ref{lemma:p} and the definition of the symmetric uncertainty box), the critical uncertainty level can be defined as
\[
u_{\mathrm{crit}} \triangleq \inf\{u\in[0,1): p_{\mathrm{worst}}(u)\ge p_{\mathrm{crit}}\}.
\]
Whenever the threshold is crossed within the admissible uncertainty range, this defines the unique minimal uncertainty level at which the certified reliable rate under the hashing-bound criterion vanishes. For all $u \ge u_{\mathrm{crit}}$, the condition $p_{\mathrm{worst}}(u) \ge p_{\mathrm{crit}}$ holds, so the certified reliable rate under the hashing-bound criterion is zero, resulting in a total loss of guaranteed payload in our framework.
\end{proof}

\subsection{A Design Map for Robust Covert Communication}
\label{subsec:design_map}

Fig.~\ref{fig:design_map} provides a \emph{design map} over the nominal channel parameters $(\eta_0,\bar n_{B,0})$. The heat colors indicate the critical uncertainty level $u_{\mathrm{crit}}$ beyond which the guaranteed payload drops to zero, obtained by solving the hashing--bound threshold in Eq.~\eqref{eq:hashing_bound_ths} for $p_{\mathrm{crit}}$ and then Eq.~(\ref{eq:p_crit}) for $u_{\mathrm{crit}}$ at each grid point.
The white contours overlay the guaranteed payload
$M_{\mathrm{rob}}(n)$ from Eq.~\eqref{eq:main_result}, evaluated at a fixed uncertainty $u=0.05 (5\%)$ using the corresponding worst-case corners $(\eta_{\min}, \bar n_{B,\min})$ for covertness and $(\eta_{\min}, \bar n_{B,\max})$ for reliability. This plot is interpreted as follows: for a chosen nominal operating point $(\eta_0,\bar n_{B,0})$, (i) the color gives the maximum tolerable relative uncertainty before the \emph{rate cliff} occurs ($R_{\mathrm{worst}}=0$), and (ii) the nearest contour gives the
guaranteed covert payload for the fixed $(n,\delta)$ stated in the caption. The earlier $8.85\%$ cliff result (see Fig.~\ref{fig:performance_vs_uncertainty}; it is also reflected in Fig.~\ref{fig:penalty_vs_uncertainty}) corresponds to one vertical slice of this map at $(\eta_0,\bar n_{B,0})=(0.9,0.12)$.

\begin{figure}[t]
  \centering
  \includegraphics[width=0.375\textwidth]{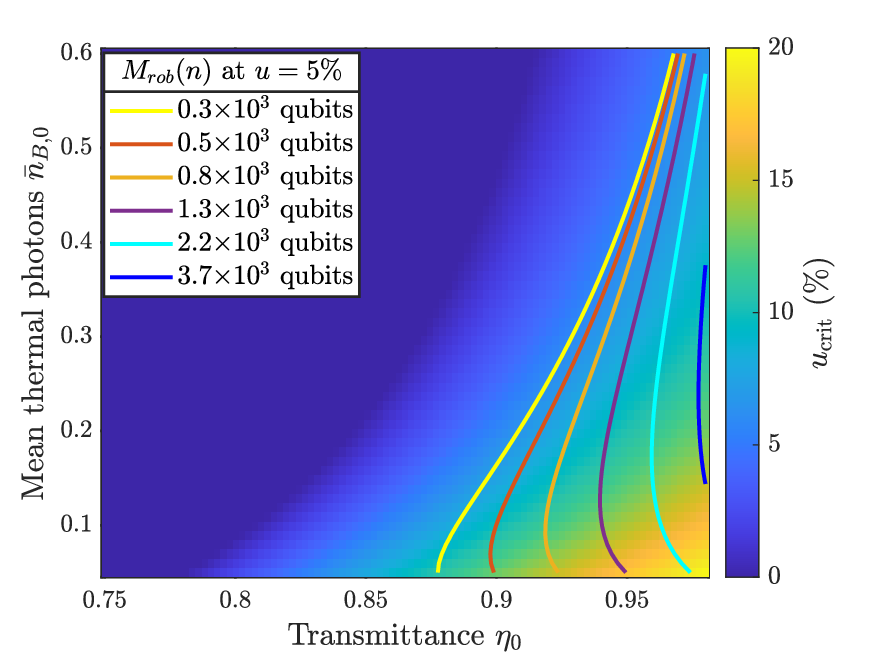}
  \caption{Design map over $(\eta_0,\bar n_{B,0})$ ($\eta_0 \ge 0.75$). Color: critical uncertainty $u_{\mathrm{crit}}$ (\%). White contours: guaranteed covert payload $M_{\mathrm{rob}}(n)$ at $u=5\%$, in units of $10^3$ qubits per block, log-spaced. Parameters: $n=10^8$, $\delta=0.05$. Higher contours imply a larger guaranteed payload.}
  \label{fig:design_map}
  \vspace{-3mm}
\end{figure}

The heatmap in Fig.~\ref{fig:design_map} reveals that only channels with high transmittance ($\eta_0\gtrsim0.9$) tolerate uncertainty meaningfully. As $\eta_0$ decreases, the design enters a sensitivity cliff where even a small mismatch can drive the reliable rate to zero. The overlaid $M_{\mathrm{rob}}(n)$ contours at $u = 5\%$ confirm that both covertness and guaranteed payload degrade sharply outside this robust region. The point $(\eta_0,\bar n_{B,0})=(0.9,0.12)$, analyzed in the 8.85\% cliff study, lies on the boundary of this robust zone.

The white contour lines in Fig.~\ref{fig:design_map} represent isolines of the guaranteed covert payload $M_{\mathrm{rob}}(n)$ at a fixed $u = 5\%$, expressed in thousands of qubits per block.
Moving toward higher contour labels (e.g., $0.5{\rm k}\!\rightarrow\!3.6{\rm k}$)
corresponds to a larger guaranteed payload under the same uncertainty level.
The colored background, in contrast, shows the critical uncertainty
$u_{\mathrm{crit}}$ beyond which the reliable rate vanishes. The intersection of high $u_{\mathrm{crit}}$ and high $M_{\mathrm{rob}}(n)$ defines the desirable robust operating region.

Interestingly, for a fixed transmittance $\eta_0$, some contours intersect a
vertical line at two distinct noise levels $\bar n_{B,0}$. This apparent
non-monotonicity arises from the competing effects of background noise:
moderate thermal noise improves covertness by increasing the masking background,
whereas excessive noise eventually destroys reliability through a larger
depolarizing probability. Consequently, the guaranteed payload
$M_{\mathrm{rob}}(n) \propto c_{\mathrm{cov}}^{\mathrm{rob}}\, R_{\mathrm{worst}}$ first increases and then decreases with $\bar n_{B,0}$, producing two operating points with an identical guaranteed payload.
\vspace{-2mm}

\section{Discussion}
\label{sec:discussion_implications}
Before discussing the numerical results, we clarify the role of the simulations used in this paper. The QuTiP-based simulations are not intended to model a specific deployed hardware platform or experimental setup. Rather, they serve to validate the Willie-side covertness component of the analytical framework by simulating the underlying four-mode bosonic channel, to verify the behavior of the divergence-based covertness metrics under finite-dimensional truncation, and to illustrate the impact of bounded channel uncertainty on the achievable covert payload when combined with the analytical hashing-bound reliability model. The physical relevance of the results is ensured through conservative parameter ranges and worked examples grounded in realistic optical link regimes, as discussed in the preceding sections. The results presented in the previous section reveal several critical features of the system, which we now discuss, beginning with the broad security implications of our findings.
\vspace{-3mm}
\subsection{Security Implications for System Design}
Our analysis provides a quantitative security framework to evaluate covert quantum communication under real-world channel uncertainty. We summarize our findings as several security implications for system designers and operators.

\paragraph{The Robust Bound as a Feasibility Test}
The robust performance bound, $M_{\mathrm{rob}}(n)$, serves as a critical design tool. For a given hardware platform, an engineer can first characterize the operational uncertainty to establish the bounds $[\eta_{\min}, \eta_{\max}]$ and $[\bar n_{B,\min}, \bar n_{B,\max}]$. Our result in Eq.~\eqref{eq:main_result} then allows them to calculate the minimum guaranteed covert payload that the system can certify for a target covertness level $\delta$. This provides a clear pass/fail test: if this guaranteed payload is insufficient for the intended application, the system is not
viable under the specified operational uncertainty.

\paragraph{The Rate Cliff Edge as a Hard Security Boundary}
The ``rate cliff edge,'' illustrated directly in Fig.~\ref{fig:performance_vs_uncertainty} and reflected in Fig.~\ref{fig:penalty_vs_uncertainty}, is not a graceful degradation, but a hard security boundary. For our nominal parameters, this threshold occurs at an uncertainty level of approximately 8.85\%. If a system is deployed in an environment where physical fluctuations could plausibly exceed this boundary, it cannot be considered robust. The covert channel effectively collapses in our framework, as a reliable link to Bob can no longer be guaranteed under the adopted worst-case criterion. This translates a physical environmental property (uncertainty) into a binary robustness outcome under the adopted worst-case criterion.

\begin{table*}[!t]
\centering
\caption{Summary of Security and Performance Guarantees}
\label{tab:guarantees}
\begin{tabular}{p{3.2cm}|p{6.0cm}|p{6.0cm}}
\hline
\textbf{Property} & \textbf{Guarantee Provided} & \textbf{Assumption / Condition} \\
\hline
\hline
\textbf{Covertness} \newline (Undetectability)
& Ensures Willie’s error probability satisfies $P_e \ge \frac{1}{2} - \delta$
& Alice uses strategy based on robust covertness constant $c_{\text{cov}}^{\text{robust}}$ \\
\hline
\textbf{Reliability} \newline (Decodability)
& Guarantees communication rate $R_{\text{worst}}$ under worst-case conditions
& Error correction matches worst-case depolarizing parameter $p_{\text{worst}}$ \\
\hline
\textbf{Robustness to Uncertainty} & Guarantees at least $M_{\mathrm{rob}}(n)$ for all parameters in the uncertainty region
& Alice knows bounds:
$[\eta_{\min}, \eta_{\max}]$,
$[\bar n_{B,\min}, \bar n_{B,\max}]$ \\
\hline
\textbf{Resistance to Passive Adversary}
& Covers adversaries with optimal quantum detection (Helstrom bound)
& Adversary is passive and non-interfering \\
\hline
\textbf{Channel Model Validity}
& Applies to the lossy thermal-noise bosonic channel model adopted in this paper
& Assumes the channel regime and reduction used in the analysis \\
\hline
\end{tabular}
\vspace{-5mm}
\end{table*}

\paragraph{The Cost of Conflict as a Fundamental Security Tax}
The conflict penalty, quantified in Fig.~\ref{fig:penalty_vs_uncertainty}, can be interpreted as a fundamental ``security tax'' that a system designer must pay to ensure robust operation in the feasible pre-cliff regime. It represents the performance lost because the system must simultaneously prepare for two different worst-case scenarios: Willie's most favorable condition for detection and Bob's most adverse condition for reception. Understanding this cost is critical, as this penalty can only be reduced by improving Alice’s operational knowledge of the channel, highlighting the critical role of sensing, calibration, or environmental control in the design of a secure communication system. Beyond the rate-cliff edge, the robust operating point itself collapses under the adopted worst-case criterion, so the post-cliff 100\% level in Fig.~\ref{fig:penalty_vs_uncertainty} should be interpreted only as an indicator of complete loss of feasible guaranteed covert communication. This underlines the practical value of the validated analytical bounds, which enable rapid exploration of system performance under uncertainty without requiring full quantum state simulation. Table~\ref{tab:guarantees} summarizes the key security and performance guarantees established by our robust framework, along with the assumptions needed for them to be valid.
\vspace{-3mm}
\subsection{Worked Examples: End-to-End Robust Guarantee}
We evaluate the robust bound for two representative operating regimes: urban nighttime free-space optical (FSO) links and short metro/in-building fiber. By evaluating the covertness and reliability corners,
$c_{\mathrm{cov}}^{\mathrm{rob}} = c_{\mathrm{cov}}(\eta_{\min}, \bar n_{B,\min})$ and
$p_{\mathrm{worst}} = p(\eta_{\min}, \bar n_{B,\max})$,
over the specified uncertainty intervals, we apply Eq.~\eqref{eq:main_result} to certify the covert payload for all channels in the uncertainty set.
More precisely, Eq.~\eqref{eq:main_result} yields $M_{\mathrm{rob}}(n)$, a guaranteed worst-case lower bound on the \emph{expected} covert payload that holds uniformly for all $(\eta,\bar n_B)\in\mathcal{U}$.
The resulting certified quantities, including $M_{\mathrm{rob}}(n)$, are summarized in Table~\ref{tab:worked_examples}. The parameter intervals for the FSO and fiber regimes are aligned with representative experimental measurements reported in the literature~\cite{Mostafa2012UrbanFSO,Alheadary2018FSOChar,Avesani2021DaylightFSO,Alimi2024FSOSurvey,Hiskett2006FiberLoss,Dynes2016MCF,verma2015high,Tripathy2024SNSPDReview}.

\begin{table}[t]
  \caption{\textbf{Worked examples under bounded uncertainty.} $n=10^8,\ \delta=0.05$. FSO: $\eta\!\in\![0.90,0.98]$, $\bar n_B\!\in\![0.02,0.12]$; Fiber: $\eta\!\in\![0.80,0.90]$, $\bar n_B\!\in\![0.001,0.02]$. Outputs: $c_{\mathrm{cov}}^{\mathrm{rob}}$, $p_{\mathrm{worst}}$, $R_{\mathrm{worst}}$, and certified $M_{\mathrm{rob}}(n)$ via \eqref{eq:main_result}.}
  \label{tab:worked_examples}
  \centering
  \vspace{0.25em}
  \begin{tabular}{p{2.6cm}cccc}
    \hline
    \textbf{Scenario} 
    & $c_{\mathrm{cov}}^{\mathrm{rob}}$ 
    & $p_{\mathrm{worst}}$ 
    & $R_{\mathrm{worst}}$ 
    & $M_{\mathrm{rob}}(n)$ \\
    \hline
    Nighttime FSO %(urban, clear)
    & 1.9144 & 0.1419 & 0.3422 & 655.0 \\
    Short fiber %(metro/in-building)
    & 0.2001 & 0.2127 & 0.1141 & 22.82 \\
    \hline
  \end{tabular}
  \vspace{0.35em}
\end{table}

In practice, uncertainty bounds can be obtained from channel probing, environmental sensing, and predictive physical models. In FSO links, for example, transmittance and background-noise bounds may be estimated from pilot measurements and radiometric sensing. Our analysis then provides the robust performance guarantee for any such estimated uncertainty region.

\section{Conclusion}
\label{sec:conclusion}

%\subsection{Summary of Contributions}
This paper introduced a robust analytical framework for covert quantum communication under bounded channel uncertainty. While prior studies established covert communication limits under idealized assumptions, our work addresses a key gap by relaxing the requirement of perfect channel state knowledge. We considered uncertainty in both transmissivity and background thermal noise (quantities that, in practice, are estimated over bounded intervals). Our analysis revealed a fundamental trade-off: the worst-case assumptions that ensure covertness are in direct conflict with those required for reliable decoding. To address this, we derived a closed-form guaranteed worst-case lower bound on the expected number of covert qubits that can be reliably and covertly transmitted across all admissible channels. This bound provides a security guarantee for covert communication in uncertain environments.

To complement the theory, we conducted QuTiP simulations of the Willie-side bosonic channel. These verified Fock-space truncation via convergence testing and confirmed close agreement for the covertness component used in the robust payload evaluation. Combined with the analytical hashing-bound reliability model, they quantified the rate penalty incurred by uncertainty, offering guidance for system designers.

Beyond point-to-point covert links, the proposed framework has direct implications for quantum networking. In networked settings, bounded uncertainty arises from dynamic routing, atmospheric conditions, and heterogeneous hardware across links. The compound-channel characterization and feasibility boundary provide a conservative criterion for link admission, covert scheduling, and resource allocation in quantum-secure networks. In particular, the collapse of the guaranteed covert payload under excessive uncertainty highlights the need for uncertainty-aware design when integrating covert communication into larger quantum-classical network infrastructures.

%\subsection{Future Work}
Several avenues remain open for future research. One important direction is to extend the framework to dynamic or stochastic channel models, where transmittance and noise parameters evolve or follow known statistical distributions. Another is the development of adaptive transmission strategies that incorporate real-time channel estimation, learning, or feedback to tighten uncertainty bounds during operation. Finally, generalizing the model to account for additional adversarial behaviors (such as active jamming or side-channel leakage)
% or covert entanglement distribution)
would further strengthen the robustness and applicability of covert quantum communication in practical deployments.
\vspace{-2mm}

\appendix
\section{Proofs}

\subsection{Proof of Lemma~\ref{lemma:ccov}}
\label{app:proofs-ccov}

\begin{proof}
To establish the monotonicity of the covertness constant $c_{\mathrm{cov}}$ in Eq.~\eqref{eq:ccov}, we examine its derivatives with respect to $\bar n_B$ and $\eta$ over the physical domain $\eta\in(0,1)$ and $\bar n_B>0$.

\paragraph{Dependence on $\bar n_B$}
Differentiating $c_{\mathrm{cov}}$ with respect to $\bar n_B$ gives
\begin{align}
\frac{\partial c_{\mathrm{cov}}}{\partial \bar n_B}
=
\frac{\sqrt{2}\,\eta\!\left(1+2\eta\bar n_B\right)}
{2(1-\eta)\sqrt{\eta\bar n_B(1+\eta\bar n_B)}}.
\end{align}
For $\eta\in(0,1)$ and $\bar n_B>0$, every factor in the numerator and denominator is strictly positive. Hence $\frac{\partial c_{\mathrm{cov}}}{\partial \bar n_B}>0$, which shows that $c_{\mathrm{cov}}$ is strictly increasing in $\bar n_B$.

\paragraph{Dependence on $\eta$}
Differentiating $c_{\mathrm{cov}}$ with respect to $\eta$ yields
\begin{align}
\frac{\partial c_{\mathrm{cov}}}{\partial \eta}
=
\frac{\sqrt{2}\,\bar n_B\!\left(1+2\eta\bar n_B+\eta\right)}
{2(\eta-1)^2\sqrt{\eta\bar n_B(1+\eta\bar n_B)}}.
\end{align}
Again, all terms are strictly positive over the domain of interest: $\bar n_B>0$, $(\eta-1)^2>0$ for $\eta\neq 1$, and the square-root term is positive. Therefore $\frac{\partial c_{\mathrm{cov}}}{\partial \eta}>0,$ so $c_{\mathrm{cov}}$ is strictly increasing in $\eta$. Thus, $c_{\mathrm{cov}}$ is monotonically increasing in both $\bar n_B$ and $\eta$ over the physical domain.
\end{proof}
\vspace{-5mm}
\subsection{Proof of Lemma~\ref{lemma:p}}
\label{app:proofs-p}

\begin{proof}
To determine the worst case for reliability, we analyze the depolarizing parameter $p(\eta, \bar n_B)=1-\frac{\eta}{\left(1+(1-\eta)\bar n_B\right)^4}$. Since $p(\eta,\bar n_B)=1-f(\eta,\bar n_B)$, its monotonicity is the reverse of that of $f(\eta,\bar n_B)=\frac{\eta}{\left(1+(1-\eta)\bar n_B\right)^4}$.

\paragraph{Dependence on $\bar n_B$}
We compute the partial derivative of $f$ with respect to $\bar n_B$:
\begin{equation}
    \frac{\partial f}{\partial \bar n_B} = \frac{-4\eta(1-\eta)}{\left(1 + (1 - \eta)\bar n_B\right)^5}.
\end{equation}
Since $\eta \in (0,1)$, the numerator is negative while the denominator is positive. Thus, $\frac{\partial f}{\partial \bar n_B} < 0$, meaning $f$ decreases with $\bar n_B$. Consequently, $p=1-f$ is a monotonically increasing function of $\bar n_B$.

\paragraph{Dependence on $\eta$}
We compute the partial derivative of $f$ with respect to $\eta$:
\begin{align}
    \frac{\partial f}{\partial \eta} = \frac{1 + \bar n_B(1+3\eta)}{\left(1 + (1 - \eta)\bar n_B\right)^5}.
\end{align}
For all valid parameters, the numerator and denominator are strictly positive. Thus, $\frac{\partial f}{\partial \eta} > 0$, so $f$ increases with $\eta$. Consequently, $p=1-f$ is a decreasing function of $\eta$.
\end{proof}
\vspace{-5mm}
\subsection{Proof of Theorem~\ref{thm:main_result}}
\label{app:Proof_thm1}
\begin{proof}
To guarantee covertness and reliability uniformly over the compound uncertainty set $ \mathcal{U}=[\eta_{\min},\eta_{\max}] \times [\bar n_{B,\min},\bar n_{B,\max}]$,
Alice must choose a single transmission probability $q$ and a single coding rate $R$ that are feasible for every realization $(\eta,\bar n_B)\in\mathcal{U}$. By Corollary~\ref{cor:robust_covertness}, covertness is guaranteed for every channel realization in $\mathcal{U}$ provided that $q \le \frac{2\delta\, c_{\mathrm{cov}}^{\mathrm{robust}}}{\sqrt{n}}$,
\[
c_{\mathrm{cov}}^{\mathrm{robust}}
=
\min_{(\eta,\bar n_B)\in\mathcal{U}} c_{\mathrm{cov}}(\eta,\bar n_B)
=
c_{\mathrm{cov}}(\eta_{\min},\bar n_{B,\min}).
\]
By Corollary~\ref{cor:worst_case_p}, a communication rate is reliably achievable for every realization in $\mathcal{U}$ provided that
\[
\begin{aligned}
R \le R_{\mathrm{worst}},
\qquad
R_{\mathrm{worst}}=\bigl[1-H(\vec p_{\mathrm{worst}})\bigr]^+,\\
p_{\mathrm{worst}}
=
\max_{(\eta,\bar n_B)\in\mathcal{U}} p(\eta,\bar n_B)
=
p(\eta_{\min},\bar n_{B,\max}).
\end{aligned}
\]

We now define the robust policy by choosing $q_{\mathrm{rob}} \triangleq \frac{2\delta\, c_{\mathrm{cov}}^{\mathrm{robust}}}{\sqrt{n}}$, and $R_{\mathrm{rob}} \triangleq R_{\mathrm{worst}}$. By construction, this policy is feasible for every $(\eta,\bar n_B)\in\mathcal{U}$.

Under the randomized transmission model, the expected number of covert qubits under this policy is
\[
\begin{aligned}
\mathbb{E}[M(n)\mid & q_{\mathrm{rob}},R_{\mathrm{rob}}]
=
n\, q_{\mathrm{rob}}\, R_{\mathrm{rob}} \\
&=
n\left(\frac{2\delta\, c_{\mathrm{cov}}^{\mathrm{robust}}}{\sqrt{n}}\right)R_{\mathrm{worst}} =
2\sqrt{n}\, c_{\mathrm{cov}}^{\mathrm{robust}}\, R_{\mathrm{worst}}\, \delta.
\end{aligned}
\]
Since this same feasible policy is used for all channel realizations in $\mathcal{U}$, and the quantity $\mathbb{E}[M(n)\mid q_{\mathrm{rob}},R_{\mathrm{rob}}]=n q_{\mathrm{rob}} R_{\mathrm{rob}}$ is independent of $(\eta,\bar n_B)$ once the policy is fixed, we obtain
\[
\begin{aligned}
M_{\mathrm{rob}}(n)
&\triangleq
\inf_{(\eta,\bar n_B)\in\mathcal{U}}
\mathbb{E}[M(n)\mid q_{\mathrm{rob}},R_{\mathrm{rob}}] \\
&=
2\sqrt{n}\, c_{\mathrm{cov}}^{\mathrm{robust}}\, R_{\mathrm{worst}}\, \delta.
\end{aligned}
\]
In particular,
\[
M_{\mathrm{rob}}(n)
\ge
2\sqrt{n}\, c_{\mathrm{cov}}^{\mathrm{robust}}\, R_{\mathrm{worst}}\, \delta,
\]
which proves the claimed guaranteed lower bound.
\end{proof}
\vspace{-5mm}
\subsection{Proof of Proposition~\ref{prop:baseline_fails}}
\label{app:Proofs-baseline_fails}

\begin{proof}
By construction, the naive policy is tuned at the nominal point $(\eta_0,\bar n_{B,0})$, so
\[
q_{\mathrm{nom}}
=
\frac{2\delta\, c_{\mathrm{cov}}(\eta_0,\bar n_{B,0})}{\sqrt{n}},
\quad
R_{\mathrm{nom}}
=
\bigl[1-H(\vec p(\eta_0,\bar n_{B,0}))\bigr]^+.
\]

Since $c_{\mathrm{cov}}(\eta,\bar n_B)$ is increasing in $\eta$ and $\bar n_B$ by Lemma~\ref{lemma:ccov}, its minimum over $\mathcal{U}_{\mathrm{base}}$ occurs at $(\eta_{\min},\bar n_{B,\min})$, and
\[
c_{\mathrm{cov}}(\eta_{\min},\bar n_{B,\min})
<
c_{\mathrm{cov}}(\eta_0,\bar n_{B,0}).
\]
Therefore,
\[
q_{\mathrm{nom}}
=
\frac{2\delta\, c_{\mathrm{cov}}(\eta_0,\bar n_{B,0})}{\sqrt{n}}
>
\frac{2\delta\, c_{\mathrm{cov}}(\eta_{\min},\bar n_{B,\min})}{\sqrt{n}},
\]
so the nominally chosen transmission probability violates the covertness requirement at $(\eta_{\min},\bar n_{B,\min})\in\mathcal{U}_{\mathrm{base}}$. Similarly, by Lemma~\ref{lemma:p}, $p(\eta,\bar n_B)$ decreases with $\eta$ and increases with $\bar n_B$, so its maximum over $\mathcal{U}_{\mathrm{base}}$ occurs at $(\eta_{\min},\bar n_{B,\max})$. Hence
\[
p(\eta_{\min},\bar n_{B,\max}) > p(\eta_0,\bar n_{B,0}).
\]
Moreover, for the depolarizing-channel hashing bound
\[
R(p)=\bigl[1-H(\vec p)\bigr]^+,
\qquad
\vec p=\left[1-\frac{3p}{4},\frac{p}{4},\frac{p}{4},\frac{p}{4}\right],
\]
the rate is monotonically nonincreasing in $p$. Indeed, on the interval where $R(p)>0$, the entropy term
\[
H(\vec p)
=
-\left(1-\frac{3p}{4}\right)\log_2\!\left(1-\frac{3p}{4}\right)
-3\left(\frac{p}{4}\right)\log_2\!\left(\frac{p}{4}\right)
\]
is monotonically increasing in $p$, and once it reaches $1$, the operator $[\cdot]^+$ keeps $R(p)$ at zero. Therefore,
\[
R_{\mathrm{nom}}
=
\bigl[1-H(\vec p(\eta_0,\bar n_{B,0}))\bigr]^+
>
\bigl[1-H(\vec p(\eta_{\min},\bar n_{B,\max}))\bigr]^+.
\]
Thus, the nominally chosen coding rate is not uniformly reliable over $\mathcal{U}_{\mathrm{base}}$. Therefore, the nominally tuned policy fails to satisfy at least one of the required constraints for some realizations in $\mathcal{U}_{\mathrm{base}}\subseteq\mathcal{U}$. By the definition of certified worst-case covert-and-reliable payload adopted in Proposition~\ref{prop:baseline_fails}, any policy that is not uniformly feasible over $\mathcal{U}$ has guaranteed payload zero. Hence $M_{\mathrm{rob,naive}}(n)=0$.
\end{proof}
\vspace{-5mm}

%\subsection{Reproduction and Availability}
%For submission, we avoid public links to preserve double-blind review. Upon acceptance, we will release a complete, versioned package on Zenodo (open access, DOI) containing: (i) code and scripts for all figures/tables, (ii) environment files (a minimal \texttt{requirements.txt} and an exact lockfile), (iii) fixed seeds and parameter grids, and (iv) step-by-step reproduction instructions and a Makefile. Code will be licensed under MIT; generated data/figures under CC~BY~4.0.%
%\unskip

\bibliographystyle{IEEEtran}
\bibliography{references}

\vfill

\end{document}